\documentclass[a4paper]{article}
\usepackage[all]{xy}\usepackage[latin1]{inputenc}        
\usepackage[dvips]{graphics,graphicx}
\usepackage{amsfonts,amssymb,amsmath,color,mathrsfs, amstext}
\usepackage{amsbsy, amsopn, amscd, amsxtra, amsthm,authblk}
\usepackage{enumerate,algorithmicx,algorithm}
\usepackage{algpseudocode}
\usepackage{upref}
\usepackage{geometry}
\usepackage[displaymath]{lineno}
\usepackage{float}
\usepackage{yhmath}
\usepackage{booktabs}
\usepackage{bm}
\usepackage{multirow}
\usepackage{makecell}
\usepackage{appendix}
\usepackage[colorlinks,
linkcolor=red,
anchorcolor=red,
citecolor=red
]{hyperref}

\numberwithin{equation}{section}

\DeclareMathOperator{\erf}{erf}
\DeclareMathOperator{\erfc}{erfc}

\newtheorem{theorem}{Theorem}
\newtheorem{lemma}{Lemma}

\newtheorem{remark}{Remark}

\begin{document}

\title{Random batch sum-of-Gaussians method for molecular dynamics simulations of particle systems}

\author[1]{Jiuyang Liang}
\author[1,2]{Zhenli Xu}
\author[1]{Qi Zhou}
\affil[1]{School of Mathematical Sciences, Shanghai Jiao Tong University, Shanghai, 200240, P. R. China}
\affil[2]{MOE-LSC, CMA-Shanghai and Shanghai Center for Applied Mathematics, Shanghai Jiao Tong University, Shanghai 200240, China}

\date{}

\maketitle

\begin{abstract}
We develop an accurate, highly efficient and scalable random batch sum-of-Gaussians (RBSOG) method for molecular dynamics simulations of systems with long-range interactions. The idea of the RBSOG method is based on a sum-of-Gaussians decomposition of the Coulomb kernel, and then a random batch importance sampling on the Fourier space is employed for approximating the summation of the Fourier expansion of the Gaussians with large bandwidths (the long-range components). The importance sampling significantly reduces the computational cost, resulting in a scalable algorithm by avoiding the use of communication-intensive fast Fourier transform. Theoretical analysis is present to demonstrate the unbiasedness of the approximate force, the controllability of variance and the weak convergence of the algorithm. The resulting method has $\mathcal{O}(N)$ complexity with low communication latency.   Accurate simulation results on both dynamical  and equilibrium properties of benchmark problems are reported to illustrate the attractive performance of the method. Simulations on parallel computing are also performed to show the high parallel efficiency. The RBSOG method can be straightforwardly extended to more general interactions with long ranged kernels, and thus is promising to construct fast algorithms of a series of molecular dynamics methods for various interacting kernels.  
   
\end{abstract}

\vspace{0.5cm}


\maketitle
%
\section{Introduction}
Molecular dynamics (MD) simulation has become one of the most popular tools for computational study for properties of nano/micro scale systems in various areas. MD furnishes kinetic and thermodynamic quantities of physical systems by the ensemble average of particle configurations produced by the integration of the Newton's equations for each particle which interacts with all other particles \cite{Allen2017ComputerLiquids,Frenkel2001Understanding}. One of the bottleneck problems in MD is the low efficiency of evaluating electrostatic forces between charged particles due to the long-range nature. Meanwhile, as the 3D domain decomposition has achieved great progresses in massively high-performance computing, the long-range nature \cite{David2011Nanoscale,French2010Rev} of the Coulomb interaction  leads to communication latency that significantly reducing the parallel efficiency.

Enormous efforts have been devoted to develop fast electrostatic solvers of MD, mostly based on the classical Ewald method \cite{Ewald1921AnnPhys} where the $1/r$ Coulomb kernel is decomposed into short-range (or near) and a smooth long-range (or far) parts. The long-range part is treated in the Fourier space, which can be effectively calculated by the lattice sums with the fast Fourier transform (FFT) acceleration, achieving an $\mathcal{O}(N\log N)$) complexity \cite{Darden1993JCP,Hockney1988Computer,essmann1995smooth}.
There are also other fast algorithms such as the treecode algorithm \cite{Barnes1986Nature}, the fast multipole method \cite{greengard1987fast}, the multigrid method \cite{trottenberg2000multigrid}, the Maxwell-equation molecular dynamics \cite{maggs2002local} and the sum-of-Gaussians (SOG)-based u-series method \cite{DEShaw2020JCP}, having been developed for MD simulations.
It is remarked that these algorithms achieve $\mathcal{O}(N)$ or $\mathcal{O}(N\log N)$ scaling, but confront low parallel scalability for large-scale simulations \cite{Arnold2013PRE} due to the intensive global communications. As an example, 
the FFT-based Ewald algorithms typically require six sequential communication rounds and expensive data-reshapes \cite{ayala2021scalability}. 
Recently, 
the random-batch Ewald (RBE) method \cite{Jin2020SISC,liang2021random} has been proposed as an alternative linear-scaling algorithm.  
The RBE avoids the use of the FFT by employing a random importance sampling from the Fourier space in calculating the long-range part of the Ewald sums.
The resulted $\mathcal{O}(N)$ cost is mathematically optimal among the Ewald-type algorithms. Besides, the stochastic nature of the RBE well addresses the scalability issue, and  excellent performance is shown for benchmark simulations of systems with size up to $10^8$ atoms \cite{liang2022superscalability}.

It is noticed that the SOG is another splitting technique to decompose the $1/r$ kernel into short-range  and long-range   parts, other than the Ewald splitting. The SOG approximates the kernel by a sum of Gaussians and the interaction ranges depend on the bandwidths of the Gaussians. The SOG approximation can be employed for many different kernels, and the method is often used in convolution integrals and kernel summations  \cite{greengard2018anisotropic,cheng1999fast,jiang2008efficient}. The SOG for the Coulomb kernel can be obtained by bilateral series approximation  \cite{beylkin2005approximation,beylkin2010approximation}, and based on it Predescu {\it et al.} \cite{DEShaw2020JCP} developed the u-series method for MD simulations, which reduces the communication cost of FFT in comparison to the Ewald-type lattice summation at the expense
of some overhead in communication bandwidth and computation, and has been implemented on the  Anton 3 supercomputer \cite{shaw2021anton}. We remark that the SOG can be potentially extended to systems with more complicated long-range interactions because of its nice feature in constructing kernel decomposition for general kernels. 
 
In this paper, we present a novel method for the calculation of long-range interactions in MD simulations, which is an extension of the RBE but with an SOG approximation for the interaction kernel. The so-called random batch SOG (RBSOG) method not only remains the superscalability of the RBE \cite{liang2022superscalability}, but also has the advantage for the applicability to more general kernels. The SOG results in a continuously differentiable approximation at the cutoff, and thus the magnitude of the truncation error is significantly decreased in comparison to the Ewald decomposition. We systematically analyze the correction on the zero-frequency mode term and the truncation error of the SOG decomposition. In the RBSOG, the near part is directly evaluated, whereas the long-range part is approximated by introducing the so-called random mini-batch sampling \cite{Jin2020SISC,jin2020random} on the Fourier space such that the use of the communication-intensive FFT is avoided, similar to the RBE. The resulting algorithm reduces the computational complexity to $\mathcal{O}(N)$, and the communication cost between cores
is also significantly reduced. Our simulations on all-atom bulk water and ionic liquid systems reveal that the spatiotemporal information for these systems on a broad range of time and length scales as well as the thermodynamical quantities are quantitatively reproduced by the RBSOG-based MD.

For another practical aspect, there is typically a substantial gap between algorithm and implementation for highly efficient computer simulations \cite{leiserson2020there}. Part of the problem is that yet tailoring an algorithm to a modern computer requires an
understanding of advanced technologies, such as parallelism, vectorization, memory caching, and saving floating-point operations. To overcome this gap, we present an optimized implementation strategy. A core-shell structure is designed for efficiently tabulating the short-range part combining with partial series expansion. A parallel sampling procedure based on message passing interface (MPI) is also developed for the long-range part vectorized via the AVX512 instructions. We implement our code into a modified version of the Large-scale Atomic/Molecular Massively Parallel Simulator (LAMMPS) \cite{plimpton1995fast,thompson2022lammps}, which is one of the mainstream MD package, by achieving the 3D domain decomposition framework \cite{thompson2022lammps}. Our numerical experiments obtain expected linear scaling cost, together with excellent performance in weak and strong scalings for parallel computing, enhancing the computational speed by about two orders of magnitude in comparison to the PPPM when $2^{11}$ cores are employed.

The paper is organized as follows. In Section \ref{model}, we overview  preliminary results on electrostatics and the kernel decomposition. In Section \ref{sec::RBSOG}, we describe the RBSOG algorithm in details. The complexity and theoretical analysis are provided in Section \ref{sec::erranal}. Section \ref{sec::numerical} contains simulation results on all-atom systems. Concluding remarks are made in Section \ref{sec::conclusion}. 

\section{Electrostatic interactions and kernel decomposition}\label{model}

Consider a charged system of $N$ particles located at $\{\bm{r}_i=(x_i,y_i,z_i),~i=1,\cdots,N\}$ with charge $\{q_{i},~i=1,\cdots,N\}$ in a cuboid domain $\Omega$ with side lengths $L_x$, $L_y$, and $L_z$, respectively. The domain is specified with a 3D periodic boundary condition such that it is replicated in all three directions so as to mimic the bulk environment. The system is assumed to be charge neutrality $\sum_{i=1}^{N}q_{i}=0$. Given the charge distribution, the potential on the $i$th particle has explicit expression of an infinite series,
\begin{equation}\label{solutionperiodic}
	\Phi_i=\sum_{j=1}^{N}\sum_{\bm{n}\in\mathbb{Z}^{3}}~ '~\dfrac{q_j}{\left|\bm{r}_j+\bm{n}\circ\bm{L}-\bm{r}_i\right|}
\end{equation}
where $\bm{L}=\left(L_x,L_y,L_z\right)$, the prime indicates that the case $i=j$ with $\bm{n}$ zero vector is excluded in the double summation, and ``$\circ$" represents the Hadamard product of two vectors. The electrostatic force of the $i$th particle is evaluated from $\bm{F}_i=-\nabla_{\bm{r}_i}U$, where $U$ is the total electrostatic energy of the system with the expression by superposition of energies on each charge that
\begin{equation}\label{energy}
	U=\dfrac{1}{2}\sum_{i=1}^{N}q_i\Phi_i.
\end{equation}
Here the coefficient $1/2$ is due to the double count of the interacting pairs.

The series \eqref{solutionperiodic} describes the electrostatic interactions between the $i$th charge and all other source and image charges. It converges conditionally \cite{Frenkel2001Understanding} due to the long-range nature, and a direct cutoff to calculate the potential is less accurate in producing physically meaningful solutions. Moreover, the Coulomb kernel has a singularity at the origin, as a result the Fourier transform cannot be directly applied. The Ewald splitting \cite{Ewald1921AnnPhys} provides a perfect solution of the two issues by dividing the Coulomb kernel into  contributions of near and far parts:
\begin{equation}\label{splitting}
\dfrac{1}{r}\doteq\mathcal{N}(r)+\mathcal{F}(r)=\frac{\erfc(\alpha r)}{r}+\frac{\erf(\alpha r)}{r},
\end{equation}
where $\erf(\cdot)$ is the error function and $\erfc(\cdot)$ is its complementary, and $\alpha$ is a positive parameter. The near part $\mathcal{N}(r)$ decays rapidly and the corresponding kernel summation problem can be truncated at a certain cutoff distance $r_c$ where the interactions beyond this distance are ignored. The far part $\mathcal{F}(r)$ is now a smooth function and decays slowly. The corresponding kernel summation is performed in the Fourier space as the Fourier series of a smooth function decays rapidly. The convergence of series \eqref{solutionperiodic} can be handled by correctly treating the zero frequency mode and the ignorance of this mode corresponds the use of the tin-foil boundary conditions. Furthermore, the FFT can be employed for accelerating the summation of the Fourier series, and one can take a large $\alpha$ such that the near part has a small interacting range, leading to an $\mathcal{O}(N\log N)$ cost in each step. These are ideas in most of popular MD packages such as LAMMPS \cite{thompson2022lammps} and GROMACS \cite{abraham2015gromacs}. 

Different from the Ewald splitting,  the SOG \cite{greengard2018anisotropic,exl2016accurate} approximates the interacting kernel by a series of Gaussians such that the near and far parts are grouped with those with small and large bandwidths, respectively. 
The SOG decomposition is particularly useful for calculating the far-part interaction as the FFT will benefit from the Gaussians which permit separation of variables. For the Coulomb kernel, one can use the integral identity for the power function, 
\begin{equation}\label{GammaExpansion}
	\dfrac{1}{r^{2\beta}}=\dfrac{1}{\Gamma(\beta)}\int_{-\infty}^{\infty}e^{-e^tr^2+\beta t}dt,
\end{equation}
with $\Gamma(\cdot)$ being the Gamma function and $\beta=1/2$. Employing variable transformation $t=\log \left(x^{2}/2\sigma^2\right)$ to the integral and applying the geometrically spaced quadrature $x=b^{-\ell}$, one obtains the following bilateral series approximation   \cite{beylkin2005approximation,beylkin2010approximation}, 
\begin{equation}\label{BSA}
	\dfrac{1}{r}\approx\dfrac{2\ln b}{\sqrt{2\pi\sigma^2}}\sum_{\ell=-\infty}^{\infty}\dfrac{1}{b^{\ell}}\exp\left[-\frac{1}{2}\left(\frac{r}{b^{\ell}\sigma}\right)^2\right],
\end{equation} 
which is an SOG expansion of $1/r$. Here, $b>1$ is a constant positive number and $\sigma$ controls the width of Gaussians. One important  feature of Eq.~\eqref{BSA} is that the relative error has the asymptotic bound as $b\rightarrow 1$ for all $r>0$ \cite{DEShaw2020JCP},
\begin{equation}\label{eq::pointwiseerror}
\left|1-\dfrac{2r\ln b}{\sqrt{2\pi\sigma^2}}\sum_{\ell=-\infty}^{\infty}\dfrac{1}{b^{\ell}}\exp\left[-\frac{1}{2}\left(\frac{r}{b^{\ell}\sigma}\right)^2\right]\right|\lesssim 2\sqrt{2}\exp\left(-\dfrac{\pi^2}{2\ln(b)}\right).
\end{equation}
It is noted that other SOG methods including least-square based methods \cite{greengard2018anisotropic,wiscombe1977exponential} and semi-analytic methods using Vall{\'e}e-Poussin sums \cite{jiuyang2021AAMM} are also developed where some of them are kernel-independent. 

The u-series  \cite{DEShaw2020JCP} constructs a fast method by making use of the far part in bilateral series approximation, i.e., the $\ell\geq0$ terms in Eq.~\eqref{BSA}. This leads to a decomposition of $1/r$ into a short-range term $\mathcal{N}_{b}^{\sigma}(r)$ and a long-range term $\mathcal{F}_{b}^{\sigma}(r)$, where
\begin{equation}\label{eq::SOGDEcomp}
	\mathcal{N}_{b}^{\sigma}(r)=\begin{cases}
	1/r-\mathcal{F}_{b}^{\sigma}(r),\quad\text{if}~r<r_c\\\\0,\qquad\qquad\quad\,\,\, \text{if}~r\geq r_c\end{cases}\quad\, 
\end{equation}
and $\mathcal{F}_{b}^{\sigma}(r)$ is an SOG expansion which takes the positive part of the bilateral series approximation \eqref{BSA} and truncates at $\ell=M$,
\begin{equation}\label{eq::SOGField}
\mathcal{F}_{b}^{\sigma}(r)=\sum_{{\ell}=0}^{M}\omega_{\ell} e^{-r^2/s_{{\ell}}^2}
\end{equation}
with coefficients,
\begin{equation}
w_{\ell}=(\pi/2)^{-1/2}b^{-\ell}\sigma^{-1}\ln b, \text{ and } s_{\ell}=\sqrt{2}b^{\ell}\sigma.	
\end{equation}
The cutoff radius $r_c$ is chosen to be the smallest root of $r\mathcal{F}_{b}^{\sigma}(r)-1$. The advantages of such a decomposition are as follows. First, the potential is exact up to the cutoff radius and it is continuous at the cutoff point where the condition $r_c\mathcal{F}_{b}^{\sigma}(r_c)-1=0$ is satisfied. 
Second, high-order continuity of the potential at $r_{c}$ could be also achieved, i.e,  the $C^1$ condition will ensure the force continuity upon the condition
\begin{equation}\label{condition2}
	\dfrac{1}{r_c^2}- \partial_r \mathcal{F}_b^{\sigma}(r_c)=0
\end{equation}
is satisfied. For fixed $b$ and $\sigma$, these continuous conditions can be conjointly reached by tuning the weight of the narrowest Gaussian to be
\begin{equation}\label{eq::omega0}
	\omega_0=\dfrac{1}{e^{-r_c^2/s_{0}^2}}\left[\dfrac{1}{r_c}-\mathcal{F}_b^{\sigma}(r_c)\right],
\end{equation}
and then solving the continuity equations to determine $r_c$. The re-definition of the narrowest Gaussian weight is necessary to  prevent large error. Due to these nice features, the u-series can produce the accuracy of the Ewald decomposition with a reduced computational effort. Moreover, the separability of the Gaussian is beneficial to save half of the sequential communication rounds if the FFT is used.

\section{Random Batch sum-of-Gaussians method}\label{sec::RBSOG}
In this section, we introduce the random mini-batch idea and apply it to electrostatic calculation based on the SOG decomposition, resulting in the RBSOG method for MD simulations.  

\subsection{Fourier expansion of the far part} \label{subsection::SOGdecomposition}
By the SOG decomposition Eq.\eqref{eq::SOGDEcomp} of Coulomb kernel, the potential energy \eqref{energy} can be decomposed as the sum of contributions from near and far parts, $U:=U_\mathcal{N}+U_\mathcal{F}$ with 
\begin{equation}\label{eq::U1}
	U_\mathcal{N}=\dfrac{1}{2}\sum_{\bm{n}}\!^\prime\sum_{i,j}q_iq_j\mathcal{N}_b^{\sigma}(|\bm{r}_{ij}+\bm{n}\circ\bm{L}|)
\end{equation}
and
\begin{equation}\label{eq::U2}
	U_\mathcal{F}=\dfrac{1}{2}\sum_{\bm{n}}\!^\prime\sum_{i,j}q_iq_j\mathcal{F}_b^{\sigma}(|\bm{r}_{ij}+\bm{n}\circ\bm{L}|)
\end{equation}
where $\bm{r}_{ij}=\bm{r}_i-\bm{r}_j$. The sum for $U_\mathcal{N}$ converges absolutely and rapidly, and one can  truncate it at $r=r_c$ in real space to simplify the computation. The sum for $U_\mathcal{F}$ converges conditionally and is treated in Fourier space. 

Let us define the Fourier transform pairs as
\begin{equation}\label{fourtrans}
	\widetilde{f}(\bm{k}):=\int_{\Omega}f(\bm{r})e^{-i\bm{k}\cdot\bm{r}}d\bm{r}\quad\,\,\text{and}\quad\,\, f(\bm{r})=\dfrac{1}{V}\sum_{\bm{k}}\widetilde{f}(\bm{k})e^{i\bm{k}\cdot\bm{r}},
\end{equation}
with $\bm{k}=2\pi (m_x/L_x,m_y/L_y,m_z/L_z)$ and $\bm{m}=(m_x,m_y,m_z)\in\mathbb{Z}^{3}$.
The structure factor $\rho(\bm{k})$ for the particle distribution is given by the conjugate of the Fourier transform of the charge density
\begin{equation}
	\rho(\bm{k}):=\sum_{i=1}^{N}q_ie^{i\bm{k}\cdot\bm{r}_i}.
\end{equation}
The 3D Fourier transform of the SOG series \eqref{eq::SOGField} reads,
\begin{equation} \label{FourierSOG}
	\widetilde{\mathcal{F}}_b^{\sigma}(\bm{k})=\pi^{3/2}\sum_{\ell=0}^{M}\omega_{\ell}s_{\ell}^3e^{-s_{\ell}^2 k^2/4},
\end{equation}
where $k=|\bm{k}|$.
The far-part energy in the Fourier space is then given by  
\begin{equation} \label{uF}
U_\mathcal{F}=U_\mathcal{F}^*+U_\mathcal{F}^0-U_\mathcal{F}^{\emph{self}},
\end{equation}
where $U_\mathcal{F}^*$ is the following series,
\begin{equation}
U_\mathcal{F}^*=\sum_{|\bm{k}|\neq0}\widetilde{\mathcal{F}}_b^{\sigma}(\bm{k})\dfrac{|\rho(\bm{k})|^2}{2V}.
\end{equation}
The second term $U_\mathcal{F}^0$ is the contribution from the zero-frequency mode, which vanishes when the tinfoil boundary condtion is specified. $U_\mathcal{F}^{\emph{self}}$ in Eq.~\eqref{uF} corresponds to the contribution from the self energies of Gaussians, 
\begin{equation}
U_\mathcal{F}^{\emph{self}}=\frac{(b-b^{-M})\ln b}{\sqrt{2\pi\sigma^2}(b-1)} \sum_{i=1}^{N}q_{i}^2.
\end{equation} 
Note that $U_\mathcal{F}^{\emph{self}}\rightarrow \sum_{i=1}^N q_{i}^2/\sqrt{2\pi\sigma^2}$ when $b\rightarrow 1$ together with $M\rightarrow \infty$.
 
 
With the Fourier expansion for the far-part energy, one can treat $U_\mathcal{N}$ and $U_\mathcal{F}$ in real and Fourier spaces, respectively. Let $r_c$ be the cutoff radius for the real space, and $\mathbb{I}(i)$ be the neighbor list of the $i$th particle, which is the set of particles within the cutoff radius.  
The force acting on the $i$th particle, $\bm{F}_i$, is composed of three contributions, 
\begin{equation}\label{eq::forcesplit}
	\begin{split}
		\bm{F}_{\mathcal{N},i}&=\sum_{j\in\mathbb{I}(i)}q_iq_j\left(\dfrac{1}{2r_{ij}^3}-\sum_{\ell=0}^{M}\dfrac{\omega_\ell}{s_{\ell}^2}e^{-r_{ij}^2/s_{\ell}^2}\right)\bm{r}_{ij},\\
		\bm{F}_{\mathcal{F},i}&=\sum_{\bm{k}\neq 0}\dfrac{q_i\bm{k}}{V}\cdot\widetilde{\mathcal{F}}_b^\sigma(\bm{k})\,\text{Im}\left(e^{-i\bm{k}\cdot\bm{r}_i}\rho(\bm{k})\right),
	\end{split}
\end{equation}
and $\bm{F}_{i}^0=-\nabla_{\bm{r}_i}U_\mathcal{F}^0$. It is noted that the self energy contribution vanishes, $\nabla_{\bm{r}_i}U_{\text{self}}=0$, as the particles are usually invariant for most physical ensembles.

Given the cutoff radius, the cost of computing the near part is proportional to the product of $N$ and the average neighbors within the volume $4\pi r_c^3$ for each particle. In the u-series method \cite{DEShaw2020JCP}, the Fourier space has a cutoff $k_c$ which is set as $\mathcal{O}(1/s_0)$, inversely proportional to the bandwidth of the narrowest Gaussian, and the grid-based FFT is employed to accelerate the calculation. Let $r_0=r_c s_0/\sqrt{2}$ which is the smallest root of $r\mathcal{F}_{b}^{1}(r)-1$. Since  the Fourier modes within the cutoff frequency is proportional $k_c^3=\mathcal{O}(1/s_0^3)$, the minimization of the total computational cost leads to the $\mathcal{O}(r_0^{3/2})$ complexity for the u-series method. For the FFT, the mesh spacing of the grid is also  proportional to $s_0$, leading to the computational effort scaling at least linearly with the number of grid points $\mathcal{O}(1/s_{0}^3)$ \cite{shan2005gaussian,deserno1998mesh}. In the following, we will introduce a random batch strategy in the Fourier space to nicely handle this problem by avoiding the use of FFT. The resulting RBSOG has linear $\mathcal{O}(N)$ complexity and the computational cost is independent of the minimal bandwidth $s_{0}$ of the SOG.

\subsection{Analysis of the zero-frequency mode}\label{subsec::infinite}
The contribution from the zero-frequency mode, $U_\mathcal{F}^0$,  is a divergent term and it needs to be properly treated to satisfy the the macroscopic property of the system. This has been discussed \cite{smith1981electrostatic,yeh1999ewald,hu2014infinite,dos2016simulations} for Ewald-type methods, but remains unexplored for the SOG-type methods.  

For the Fourier transform, one performs the Taylor series expansion of $U_\mathcal{F}^0$ with respect to $\bm{k}$, takes the $\bm{k}\rightarrow \bm{0}$ limit, and then obtains,
\begin{equation}\label{eqeq::UIB}
		U_\mathcal{F}^0=\dfrac{\pi^{3/2}}{2V}\lim_{\bm{k}\rightarrow\bm{0}}\sum_{i,j}q_iq_j\sum_{\ell=0}^M w_{\ell}s_{\ell}^3\left[1-s_{\ell}^2|\bm{k}|^2/4+i\bm{k}\cdot\bm{r}_{ij}-\dfrac{1}{2}(\bm{k}\cdot\bm{r}_{i,j})^2+\mathcal{O}(|\bm{k}|^3)\right].
\end{equation}
In the summation, the first two terms vanishes due to the charge neutrality. It is noted that even for a non-neutral system (e.g., a quasi-2D system with surface charge being treated implicitly), these terms does not depend on $\bm{r}$ and thus can be renormalized. The third term  in Eq.\eqref{eqeq::UIB} vanishes too due to the symmetry condition, $\bm{k}\cdot\bm{r}_{ij}=-\bm{k}\cdot\bm{r}_{ji}$. One then has,
\begin{equation}
	U_\mathcal{F}^0=-\dfrac{1}{4V}\lim_{\bm{k}\rightarrow\bm{0}}\sum_{i,j}q_iq_j(\bm{k}\cdot\bm{r}_{ij})^2\widetilde{\mathcal{F}}_b^\sigma(0),
\end{equation}
which requires the knowledge of factor $\widetilde{\mathcal{F}}_b^\sigma(0)$. 

The Fourier transform of the SOG $\mathcal{F}_b^{\sigma}(r)$ can be written as, 
\begin{equation}\label{eqeq::B3}
	\begin{split}
		\widetilde{\mathcal{F}}_b^{\sigma}(\bm{k})
		&=\dfrac{4\pi}{k^2}-\int_{\Omega_{r_c}}\left[\dfrac{1}{\bm{r}}-\mathcal{F}_b^{\sigma}(|\bm{r}|)\right]e^{-i\bm{k}\cdot\bm{r}}d\bm{r},
	\end{split}
\end{equation}
where $\Omega_{r_c}$ is the ball of radius $r_c$ centered at the origin, and the $4\pi/k^2$ term is the Fourier transform of Coulomb kernel. Clearly, the asymptotics of $\widetilde{\mathcal{F}}_b^{\sigma}(0)$ for $\bm{k}\rightarrow \bm{0}$ is $4\pi/k^2$, thus,   
\begin{equation}\label{eqeq::UUIB}
	U_\mathcal{F}^0 =-\dfrac{\pi}{V}\sum_{i,j}q_iq_j\lim_{\bm{k}\rightarrow\bm{0}}\dfrac{(\bm{k}\cdot\bm{r}_{ij})^2}{k^2}.
\end{equation}

Eq.\eqref{eqeq::UUIB} is consistent with that for the Ewald summation \cite{hu2014infinite,dos2016simulations}. This is not a surprise as the $\bm{k}\rightarrow\bm{0}$ term accounts for the long-rang electrostatic correlations at the limit $|\bm{r}|\rightarrow\infty$. Since the SOG series with $M\rightarrow \infty$ is exact for the far-field limit, it naturally gives this consistency with the Ewald splitting. 

If the tinfoil boundary conditions is specified for $r\rightarrow \infty$, the dielectric permittivity (throughout the paper it sets to be 1) becomes infinity and $U_\mathcal{F}^0$ vanishes. In the case of finite permitivity, one can assume a infinitely large crystal box built up along spherical radial distance by filling the azimuthal angle $\theta$ and the polar angle $\varphi$ parts for each $r$. A spherical infinite boundary term \cite{hu2014infinite} can be obtained,
\begin{equation}\label{eq::sperical}
	\begin{split}
		U_\mathcal{F}^0&=-\dfrac{\pi}{V}\sum_{i,j}q_iq_j\lim_{\bm{k}\rightarrow\bm{0}}\dfrac{1}{4\pi k^2} \int_{0}^{2\pi}\int_{0}^{\pi}(\bm{k}\cdot\bm{r}_{ij})^2\sin\theta d\theta d\varphi\\
		&=\dfrac{2\pi}{3V}(\mathcal{M}_x^2+\mathcal{M}_y^2+\mathcal{M}_z^2),
	\end{split}
\end{equation}
where $(\mathcal{M}_x,\mathcal{M}_y,\mathcal{M}_z)=\sum_i q_i \bm{r}_i$ is the dipole moment. This zero mode contribution corresponds to the case of the periodic system embedded in a medium with finite dielectric permittivity.  

Now suppose that one considers a system of a slab geometry with 2D periodicity along the $xy$ plane, where a planar infinite boundary term along the $z$ direction is often required. One can apply the conditions $k_x = k_y = 0$ and take the limit of $k_z=0$. This yields,
\begin{equation} \label{planar}
U_\mathcal{F}^0=-\dfrac{\pi}{V}\sum_{i,j}q_iq_j\lim_{k_z\rightarrow0}\left[\lim_{k_x,k_y\rightarrow 0}\dfrac{(\bm{k}\cdot\bm{r}_{ij})^2}{|\bm{k}|^2}\right]=\dfrac{2\pi}{V}\mathcal{M}_z^2.
\end{equation}
This dipole-term contribuiton has been used for the Ewald3DC \cite{yeh1999ewald} to handle electrostatics in charged systems with slab geometry. For other geometries, one can refer to Refs. \cite{smith1981electrostatic,dos2016simulations}. We remark that our discussion on the RBSOG will focus on 3D periodic systems, but  the algorithm can be easily  extended to systems with other boundary conditions by taking into account the correct correction from the zero-frequency mode energy. For example, if the system is partially periodic in some directions with dielectric interfaces \cite{maxian2021fast,liang2020harmonic} which is often considered for nanopores and 2D materials, the extension of our algorithm is straightforward by using the planar infinite boundary term \eqref{planar}.
 
\subsection{Random batch importance sampling}\label{subsec::randombatch}
We now work on conducting a fast method for evaluating the Fourier space force with $\mathcal{O}(N)$ complexity and less communication cost. The basic idea is to use a random batch importance sampling to approximate the force $\bm{F}_{\mathcal{F},i}$ by a modification of the sampling for the RBE \cite{Jin2020SISC}.

The Fourier transform  $\widetilde{\mathcal{F}}_b^{\sigma}(\bm{k})$ in Eq. \eqref{FourierSOG} is also an SOG series, which is summable and can be normalized to a discrete probability distribution. Denote the sum of such factors by
\begin{equation} \label{eq::S}
S:=\pi^{-\frac{3}{2}}\sum_{\bm{k}\neq \bm{0}}\widetilde{\mathcal{F}}_b^{\sigma}(\bm{k})=\sum_{\ell=0}^{M}\omega_{\ell}s_{\ell}^3\left(\prod_d H_{\ell}^d-1\right)
\end{equation} 
with
\begin{equation}\label{Hell}
H_{\ell}^d:=\sum_{m\in\mathbb{Z}}e^{-s_{\ell}^2\pi^2m^2/L_d^2}=\sqrt{\dfrac{L_d^2}{\pi s_{\ell}^2}}\sum_{m\in\mathbb{Z}}e^{-\pi^2m^2L_d^2/s_{\ell}^2}
\end{equation}
where $d\in\{x,y,z\}$ denotes the index of the Cartesian coordinate, and the second equality in Eq.\eqref{Hell} holds due to the Poisson summation formula. Eq.\eqref{Hell} can be simply truncated at some $m=\pm \mathfrak{m}$ such that $\mathfrak{m} L_d/s_{\ell}=\mathcal{O}(1)$ to obtain an easier-to-calculate form (generally speaking, $\mathfrak{m}=2$ is enough). Then, one can regard the sum as an expectation over the probability distribution
\begin{equation}\label{distribution}
\mathscr{P}_{\bm{k}}:= \frac{\widetilde{\mathcal{F}}_b^{\sigma}(\bm{k})}{\pi^{3/2} S}.
\end{equation}
The distribution Eq.\eqref{distribution} is a summation of discrete Gaussian distributions. Compared to the case of the RBE \cite{Jin2020SISC}, $\mathscr{P}_{\bm{k}}$ is by no means separable when $M>1$, whereas it is still summable and thus can be efficiently sampled by strategy  below through an acceptance-rejection criteria.

We apply the Metropolis--Hastings (MH) \cite{hastings1970monte} algorithm to sample from the discrete distribution Eq.\eqref{distribution}. In the MH procedure, the proposal $\bm{m}^*=(m^{*}_x,m^{*}_y,m^{*}_z)$ is generated by first drawing
\begin{equation}
m^{*}_d\sim \mathcal{N}\left(0,~ \frac{1}{2}(L_d/s_{0}\pi)^2\right),
\end{equation}
from the normal distribution with mean zero and variance $(L_d/s_{0}\pi)^2/2$, separately. By choosing $m^{*}_d$ as the normal distribution corresponding to the widest Gaussian rather than other Gaussians, the advantages of importance sampling are threefold. First, $m^{*}_d$ is separable and easy to sampling. Second, since the long wave modes are more important for the periodic effects and are more likely to be chosen, this importance sampling strategy is efficient and accurate by the measurement of the variance reduction. Third, the high frequency modes have also decent probability to be selected since the widest Gaussian has slowest decaying rate.

The new sample $\bm{m}^*$ and its corresponding acceptance probability $q(\bm{m}^*|\bm{m})$ are denoted by following the procedure in \cite{Jin2020SISC}, expressed as,
\begin{equation}
\bm{m}^{*}=\text{round}(m^{*}_x,m^{*}_y,m^{*}_z)
\end{equation}
and
\begin{equation}
q(\bm{m}^{*}|\bm{m})=\prod_d q(m^{*}_d|m_d)
\end{equation}
where
\begin{equation}
\begin{split}
q(m^{*}_d|m_d)&= \dfrac{\sqrt{\pi} s_0}{L_d}  \int_{m^{*}_d-1/2}^{m^{*}_d+1/2}e^{-(\pi s_0m /L_d)^2}dm\\
&=
\begin{cases}
\erf\left(\dfrac{\pi s_{0}}{2L_d}\right), &m^{*}_d=0,\\\\
\dfrac{1}{2}\left[\erf\left(\dfrac{(m^{*}_d+1/2)\pi s_{0}}{2L_d}\right)-\erf\left(\dfrac{(m^{*}_d-1/2)\pi s_{0}}{2L_d}\right)\right], &m^{*}_d\neq 0.
\end{cases}
\end{split}
\end{equation}
The acceptance rate is expected to be appropriate due to
\begin{equation}
\dfrac{p(\bm{m}^{*})}{p(\bm{m})} \approx \dfrac{q(\bm{m}^{*}|\bm{m})}{q(\bm{m}|\bm{m}^{*})}.
\end{equation}
Note that the samples with $m_x=m_y=m_z=0$ will be discarded. In Section \ref{subsec::implementation}, we develop a useful parallel strategy such that the sampling procedure can be efficiently performed.

The MD simulation can then be done via this random batch importance sampling strategy. At each simulation step, one picks a number of batches $P$ and draws $P$ frequencies $\{\bm{k}_{\eta},\eta=1,\cdots,P\}$ from the discrete distribution $\mathscr{P}_{\bm{k}}$ using the MH procedure. Then the far-field force $\bm{F}_{\mathcal{F},i}$ can be approximated by the following random variable 
\begin{equation}\label{eq::approximateForce}
\bm{F}_{\mathcal{F},i}^*:=\dfrac{S}{P}\sum_{\eta=1}^{P}\dfrac{\pi^{3/2}q_i\bm{k}_{\eta}}{V}\text{Im}\left(e^{-\mathrm{i}\bm{k}_{\eta}\cdot\bm{r}_i}\rho(\bm{k}_{\eta})\right).
\end{equation}
In MD simulation, we use this stochastic force $\bm{F}_{\mathcal{F},i}^*$ to conduct simulation rather than $\bm{F}_{\mathcal{F},i}$, resulting in a cheaper version with computational complexity $\mathcal{O}(PN)$ in comparison to lattice-based Ewald-type methods. This implies that the RBSOG method has linear complexity per timestep as $P=\mathcal{O}(1)$. We prove that $\bm{F}_{\mathcal{F},i}^{*}$ is an unbiased estimator of $\bm{F}_{\mathcal{F},i}$ with bounded variance in Section \ref{subsec::consistency}. The molecular dynamics method \cite{Frenkel2001Understanding} using the RBSOG is summarized in Algorithm \ref{al::RBSOG}. 

\begin{algorithm}[H]
	\caption{(Random batch sum-of-Gaussians algorithm)}\label{al::RBSOG}
	\begin{algorithmic}[1]
		\State Choose $r_c$ (the cutoffs in real space), $\Delta t$ (time step size), total simulation step $N_{T}$, and batch size $P$. Construct applicable SOG decomposition as in Eq.\eqref{eq::SOGDEcomp}.  Initialize positions and velocities of charges of all particles.
		\State Sample sufficient number of nonzero $\bm{k}\sim \mathscr{P}_{\bm{k}}$ by the MH procedure to form a frequency sequence.
		\For {$n \text{ in } 1: N_{T}$}
		\State Integrate Newton's equations for time $\Delta t$ with appropriate integration scheme and some appropriate thermostat. The real part of the Coulomb forces is computed using $\bm{F}_{\mathcal{N},i}$ as in Eq.\eqref{eq::forcesplit}. The Fourier part is then computed using the stochastic force $\bm{F}_{\mathcal{F},i}^{*}$ as in Eq.\eqref{eq::approximateForce} with the $P$ frequencies selected from the frequency set in order.
		\EndFor
	\end{algorithmic}
\end{algorithm}

Note that the SOG decomposition in Eqs.~\eqref{eq::SOGDEcomp} and \eqref{eq::SOGField} can be also calculated by a lattice-based FFT \cite{DEShaw2020JCP}. In comparison to this method, the RBSOG algorithm has two merits. First, the use of the random batch idea combined with importance sampling avoids the communication-intensive framework FFT. Second, it is a mesh-free algorithm, and the computational cost for evaluating multiple Gaussians is significantly reduced to $\mathcal{O}(MP)$, independent of the particle number. For comparison, the FFT leads to an on-grid convolution of Gaussians, which is computed in real space with $\mathcal{O}(MN)$ operations.

For the sampling of the Fourier modes, one can also treat the SOG $\mathcal{F}_{b}^{\sigma}(r)$ as the sum of discrete Gaussian potentials $\omega_{\ell}e^{-r^2/s_{\ell}^2}$, and each can be normalized to a discrete probability distribution. One can build random mini-batch sampling proceduce for these Gaussian potentials independently with adaptive batch sizes depending on the Gaussian bandwidths.  This is intuitive, but less efficient because the sampling processes are needed for all Gaussians.

\subsection{Implementation details}\label{subsec::implementation}
We present the optimized implementation details with distributed-memory parallelism and vectorization for the RBSOG-based MD simulations. Our implementation is based on the message passing interface (MPI) and the Intel 512-bit single-instruction multiple-data (SIMD) instruction. Below we discuss the strategies for the short-range and long-range components, respectively.   

\begin{figure*}[ht]	
	\centering
	\includegraphics[width=1.0\textwidth]{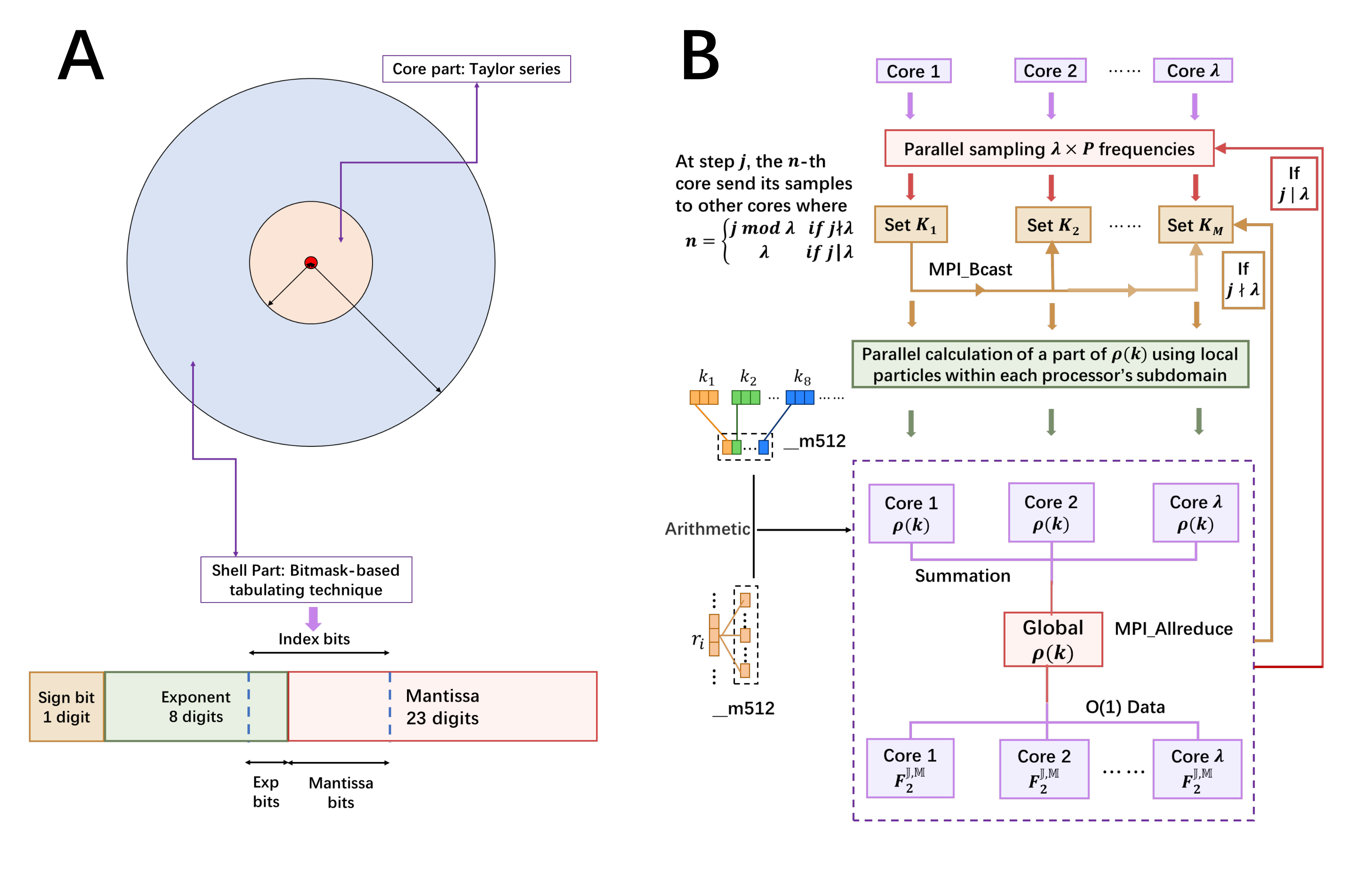}
	\caption{Sketch map. (A) The core-shell structured kernel approximation technique applied to the calculation of short-range part. (B) Parallel strategy in the Fourier space employing the SIMD.}
	\label{fig:sketchmap}
\end{figure*}


The near part $\bm{F}_{\mathcal{N},i}$ in Eq.\eqref{eq::forcesplit} requires the computation of $M$ Gaussians for each target-source pair. This is expensive from the point of view of computational efficiency and we should avoid direct calculation of Gaussians for each step, as well as the square root operation. Here, we introduce a core-shell structured kernel approximation by introducing an additional cutoff radius $r_{\text{in}}$ which is smaller than $r_c$, as is depicted in Fig.\ref{fig:sketchmap} (A).  We treat the core neighbors by $\mathscr{Q}-$terms partial Taylor expansion, i.e., the Gaussians are approximated by Taylor expansion whereas the square root and reciprocal are directly computed,
\begin{equation}
\dfrac{1}{2r^3}-\sum_{\ell=0}^{M}\dfrac{\omega_{\ell}}{s_{\ell}^2}e^{-r^2/s_{\ell}^2}= \dfrac{1}{2r^3}-\sum_{i=1}^{\mathscr{Q}}\mathcal{A}_{i}r^{2i-2}+\mathcal{O}\left(\mathcal{A}_{\mathscr{Q}+1}r^{2\mathscr{Q}}\right),
\end{equation} 
with
\begin{equation}
\mathcal{A}_{i}=\sum_{\ell=0}^{M}\dfrac{\omega_{\ell}}{(i-1)!s_{\ell}^2}\left(\dfrac{1}{s_{\ell}^2}\right)^{i-1}
\end{equation}
being the precomputed Taylor coefficients. Since the Gaussian is smooth near the origin, a small number of Taylor terms will give sufficient precision for a moderate $r_{\text{in}}$. 

The shell neighbors (particles between $r_\mathrm{in}$ and $r_c$) are treated by the bitmask-based tabulating technique \cite{wolff1999tabulated}. The expensive square root operation is avoided by using $r_{ij}^2$ as the measure of choosing shell neighbors. Note that the binary representation of a single float-point number typically contains a sign bit, $8$ exponent bits, and a mantissa component of $23$ bits. We take a few low order bits ($B_{\text{exp}}$) from the exponent and some high order bits ($B_{\text{man}}$) from the mantissa.  Fig.\ref{fig:sketchmap} (A) depicts each component of the index bits. The total table size is $2^{B_{\text{exp}}+B_{\text{man}}}$. A linear interpolation is used to approximate the data between successive points in the table.

Here are two remarks for the tabulating technique. First, the choice of table length is a tradeoff between accuracy and speed. A larger size provides more accurate force calculations, but requires more memory which can slow down the simulation. The analysis on the local and accumulated errors is referred to \cite{wolff1999tabulated,andrea1983role}. Second, the tabulated error will increase when the kernel tends to be singular, but is well solved by setting a core region.  


We next describe the sampling and domain decomposion approaches for the far part. Each MD step requires a serial importance sampling procedure and a global broadcast operation, and this cost can be eliminated by the designed non-jammed communication and computation/communication overlapping. 
Fig.\ref{fig:sketchmap}(B) describes the following procedure.
Suppose that $\lambda$ MPI ranks are employed.  The $\lambda$ independent sampling processes are first executed in parallel within each rank. Then the $1$-st MPI rank broadcasts the samples to other ranks using blocking operation. And then, the computation step of the Coulomb interaction is executed, whereas the samples in the $2$-nd MPI rank is concurrently broadcasted. The above two steps are then repeated for $\lambda-1$ times followed by a new sampling loop. This sampling strategy evaluates and updates the samples every $\lambda$ steps. It is easy to implement, but significantly improves the scalability.  

On the approximation of the far-field force by Eq.\eqref{eq::approximateForce}, the samples and the particle positions are packaged into $512$-bit vectors when the structure factors $\rho(\bm{k})$ are evaluated using the local atoms of each MPI rank. Only one global operation, MPI$\_$Allreduce, is required for all the structure factors. The approximated force $\bm{F}_{\mathcal{F},i}^{*}$ of each particle is then obtained from the structure factors. For better demonstration of the performance of RBSOG, we also implement the state-of-the-art techniques \cite{thompson2022lammps} including the domain decomposition, the ghost-atom communication, and the construction of neighbor lists combined with a multiple-page data structure. 

\section{Error analysis}\label{sec::erranal}
In this section, we provide some error analysis on the RBSOG method, including the error estimate of the truncated Gaussian series, the error in the near-field energy, and the artificial variance in the random mini-batch approximation.  

\subsection{Truncation errors and parameter determination} \label{subsec::error}

To estimate the errors of the near- and far-field energies   \eqref{eq::U1} and \eqref{uF},  we introduce the relative errors by
\begin{equation}
\mathscr{R}^{\mathcal{N}}(M):=\left|U_{\text{err}}^{\mathcal{N}}/U_{\mathcal{N}}\right|	,\quad\text{and}\quad\mathscr{R}^{\mathcal{F}}(M):=\left|U_{\text{err}}^{\mathcal{F}}/U_{\mathcal{F}}^{*}\right|,
\end{equation}
where $U_{\text{err}}^{\mathcal{N}}$ is the truncation error due to the cutoff of interactions in real space, 
\begin{equation}\label{eq::UerrN}
	U_{\text{err}}^{\mathcal{N}}=\dfrac{1}{2}\sum_{|\bm{r}_{ij}+\bm{n}\circ\bm{L}|>r_c}q_iq_j\mathcal{N}_b^{\sigma}(|\bm{r}_{ij}+\bm{n}\circ\bm{L}|),
\end{equation}
and $U_{\text{err}}^{\mathcal{F}}$ is the error by truncating the infinite SOG for the Fourier series,
\begin{equation}\label{eq:UerrF}
	U_{\text{err}}^{\mathcal{F}}=\sum_{|\bm{k}|\neq0}\left[\lim_{M\rightarrow\infty}\widetilde{\mathcal{F}}_{b}^{\sigma}(|\bm{k}|)-\widetilde{\mathcal{F}}_{b}^{\sigma}(|\bm{k}|)\right]\dfrac{|\rho(\bm{k})|^2}{2V},
\end{equation}
where $\widetilde{\mathcal{F}}_{b}^{\sigma}(|\bm{k}|)$ is defined via Eq.\eqref{eqeq::B3}.
Theorem \ref{theorem3} presents an estimate for $\mathscr{R}^{\mathcal{F}}(M)$, which is from  Ref.~\cite{DEShaw2020JCP}, but with a slight difference as  the coefficient $\omega_{0}$ here is modified to satisfy the continuous conditions by Eq.~\eqref{eq::omega0}.

\begin{theorem}\label{theorem3}
	Given parameters $b$ and $\sigma$ for the SOG decomposition as \eqref{eq::SOGDEcomp}, the error  $U_{\text{err}}^{\mathcal{F}}$ due to the truncation of the SOG at $\ell=M$ can be estimated by
\begin{equation}
	\mathscr{R}^{\mathcal{F}}(M)\leq b^{2M}e^{-2(b^{2M}-1)(\pi b\sigma/\mathcal{L})^2}
\end{equation}
with $\mathcal{L}=\max\{L_x,L_y,L_z\}$ being the maximal edge of the simulation box.
\end{theorem}

\begin{proof}
	Note that the parameters of Gaussians have the recursive relations $\omega_{\ell+1}=b^{-1}\omega_{\ell}$ and $s_{\ell+1}=bs_\ell$ for $\ell\geq 0$, respectively. By Eq.~\eqref{eq:UerrF},  an upper bound of the error for truncating the SOG series at $\ell=M$ is, 
	\begin{equation}\label{eqeq::inequ}
		\begin{split}
			\left|U_{\text{err}}^{\mathcal{F}}\right|
			=&\sum_{|\bm{k}|\neq0}\pi^{3/2}\dfrac{|\rho(\bm{k})|^2}{2V}\sum_{\ell=M+1}^{\infty}\omega_{\ell}s_{\ell}^3e^{-s_{\ell}^2|\bm{k}|^2/4}\\
			=&\dfrac{b^{2M}}{2V}\sum_{|\bm{k}|\neq0}\pi^{3/2}\dfrac{|\rho(\bm{k})|^2}{2V}\sum_{\ell=1}^{\infty}\omega_{\ell}s_{\ell}^3e^{-s_{\ell}^2|\bm{k}|^2/4}e^{-(b^{2M}-1)s_{\ell}^2|\bm{k}|^2/4}\\
			\leq& b^{2M}e^{-2(b^{2M}-1)(\pi b \sigma/\mathcal{L})^2}\left|U_{\mathcal{F}}^{*}\right|,
		\end{split}
	\end{equation}
	where the second identity is due to the use of the recursive relations with the index $\ell$ switching from $M+1$ to 0, and the inequality follows from $|\bm{k}|\geq 2\pi/\mathcal{L}$ and $\ell\geq 1$. 
\end{proof}

\begin{figure*}[ht]	
	\centering
	\includegraphics[width=0.7\textwidth]{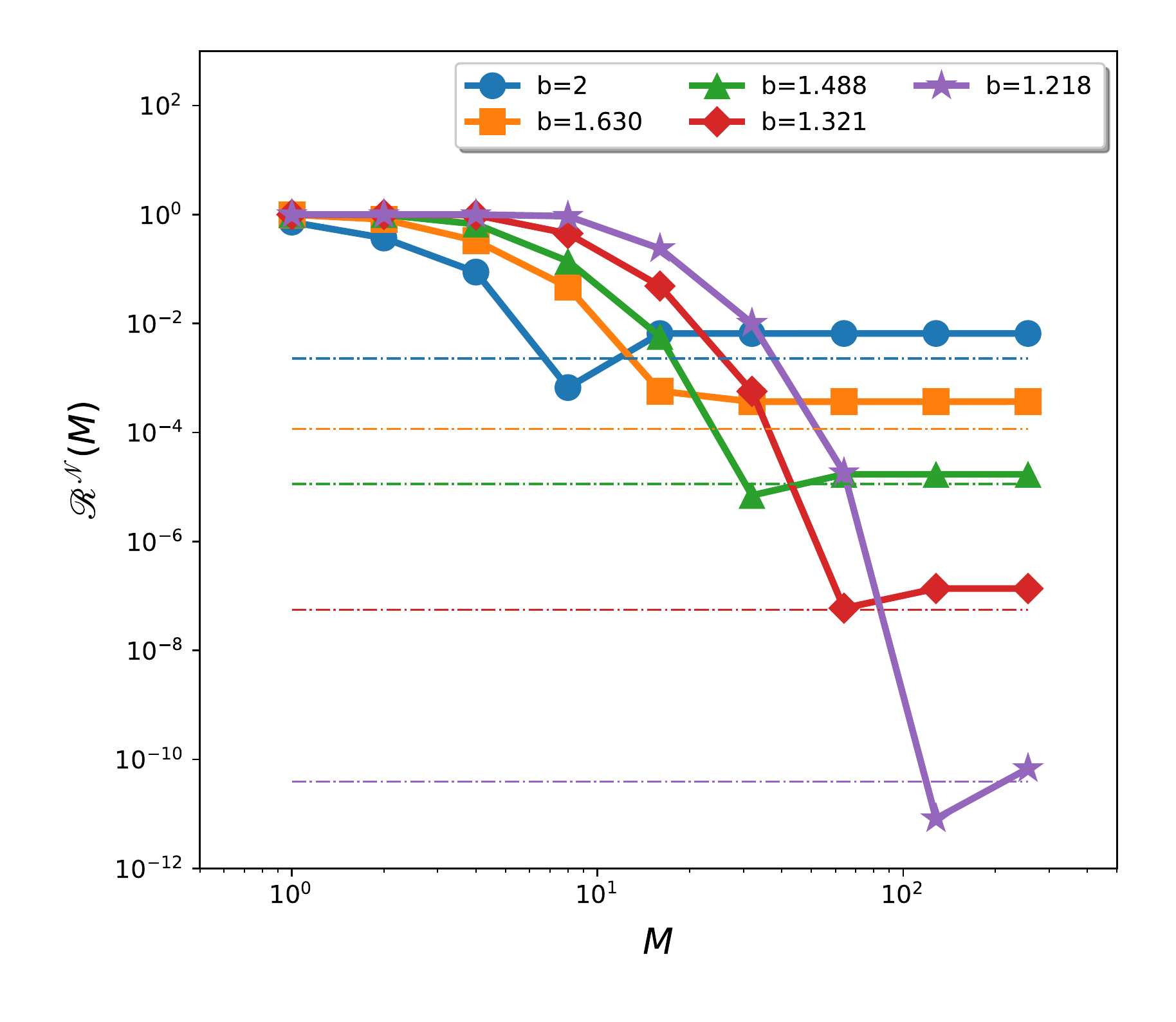}
	\caption{ $U_{\text{err}}^{\mathcal{N}}$ as a function of the number of truncated terms, $M$, for different $b$. Here $r_c$ is fixed as $1nm$ and the corresponding $\sigma$ are shown in Table \ref{tabl:parameter}. The dash-dotted lines indicate the corresponding pointwise error given in Eq.\eqref{eq::pointwiseerror}.}
	\label{fig:madelung}
\end{figure*}

For the near-field truncation error $U_{\text{err}}^{\mathcal{N}}$, a theoretical estimate remains open. It is known in the liquid-state theory \cite{parola1995liquid,weeks1971role} that the fluctuation in the long-range force is not sensitive to affect the dynamical properties. Thus, one would expect that the relative error $\mathscr{R}^{\mathcal{N}}(M)$ will be at the same level as the uniformly bounded pointwise error by Eq.\eqref{eq::pointwiseerror} upon $M$ large enough. To validate it,  we conduct numerical calculation of a NaCl cubic-like crystalline lattice with lattice size $4nm$, where the periodic images can be accurately calculated until convergence. The results are displayed in Fig. \ref{fig:madelung}. Here $r_c=1nm$ is set.  The results clearly illustrate the following error estimate:   
\begin{equation}\label{eq::errorshort} \lim_{M\rightarrow\infty}\mathscr{R}^{\mathcal{N}}(M)\approx2\sqrt{2}e^{-\pi^2/2\ln b}
\end{equation}
in ageement with the pointwise error \eqref{eq::pointwiseerror}.

One can determine the parameters of the SOG decomposition by the error estimates \eqref{eqeq::inequ} and \eqref{eq::errorshort} using the following procedure. Let the cutoff radius $r_c$ and the error tolerance $\varepsilon_{\text{rel}}$ be given apriori. Select $b$ and $\sigma$ such that  conditions $r_c\mathcal{F}_{b}^{\sigma}(r_c)-1=0$ and Eq.\eqref{condition2} are satisfied, where $b$ also satisfies Eq. \eqref{eq::errorshort}. The truncation number $M$ for the SOG is chosen  such that $
\max\left\{\mathscr{R}^{\mathcal{N}}(M),~\mathscr{R}^{\mathcal{F}}(M)\right\}\leq \varepsilon_{\text{rel}}.
$
Table \ref{tabl:parameter} shows the errors for five groups of parameters with respect to $r_c=1nm$ which is often used for practical simulations. We remark that the choice of $M$ in Table \ref{tabl:parameter} is slightly larger than those in Ref.~\cite{DEShaw2020JCP}, since we also consider the influence of $U_{\text{err}}^{\mathcal{N}}$. One shall also note that a large $M$ has only minor effects on the algorithm efficiency in both real space and Fourier space by using advanced implementation techniques described in Section \ref{subsec::implementation}.

\renewcommand\arraystretch{1.4}
\begin{table}[!ht]
	\caption{Parameter sets for the SOG decomposition with $C^{1}$ continuity of the kernel at $r_c=1nm$. $M$ is the minimum number of terms satisfying the error criteria. }
	\centering
	\begin{tabular}{ccccc}
		\hline
		$b$&$\sigma$&$w_0$& $\varepsilon_{\text{rel}}$&$M$ \\\hline
		$2$&$5.027010924194599$&$0.994446492762232252$&$2.289e-3$ & $6$\\\hline
		$1.62976708826776469$&$3.633717409009413$&$1.00780697934380681$ &$1.158e-4$& $16$\\\hline
		$1.48783512395703226$&$2.662784519725113$&$0.991911705759818$&$1.142e-5$&$30$\\\hline
		$1.32070036405934420$&$2.277149356440992$&$1.00188914114811980$ &$5.583e-8$& $64$ \\\hline
		$1.21812525709410644$&$1.774456369233284$&$1.00090146156033341$ &$3.389e-11$& $102$ \\\hline
	\end{tabular}
	\label{tabl:parameter}
\end{table}

\subsection{Convergence and algorithm complexity} \label{subsec::consistency}
We provide some theoretic evidence for the convergence and complexity of the RBSOG algorithm. We will prove that the MD simulation using stochastic force $\bm{F}_{\mathcal{F},i}^{*}$ by Eq. \eqref{eq::approximateForce} can well approximate the results of using exact force $\bm{F}_{\mathcal{F},i}$.

Denote the fluctuation of the random batch approximation for the Fourier part of the force on particle $i$ by
\begin{equation}
	\bm{\Xi}_i=\bm{F}_{\mathcal{F},i}^{*}-\bm{F}_{\mathcal{F},i}.
\end{equation}
By a direct calculation of its expactation and valance, one has the following Lemma \ref{forcebias}.

\begin{lemma}\label{forcebias}
	The fluctuation in force $\bm{\Xi}_i$ is unbiased, i.e., 
	$	\mathbb{E}\boldsymbol{\Xi_i}=\boldsymbol{0},$
	and the variance is expressed by,
	\begin{equation}\label{pro2}
		\mathbb{E}|\bm{\Xi}_i|^2=\frac{1}{P}\left[\frac{4\pi^3q_i^2S|\boldsymbol{k}|^2}{V^2}\sum_{\bm{k}\neq 0}\widetilde{\mathcal{F}}_{b}^{\sigma}(|\bm{k}|)\left|\mathrm{Im}\left(e^{-i\boldsymbol{k}\cdot\boldsymbol{r_i}}\cdot\rho(\boldsymbol{k})\right)\right|^2 -\left|\bm{F}_{\mathcal{F},i}  \right|^2\right].
	\end{equation}
\end{lemma}
The unbiased property in Lemma \ref{forcebias} implies consistency of the random batch sampling. Eq.\eqref{pro2} illustrates that the variance of $\bm{F}_{\mathcal{F},i}^{*}$ scales as $\mathcal{O}(P^{-1})$. More precisely, we have the following theorem under the mean-field assumption.
\begin{theorem}\label{theorem::variance} Let $\rho_r=N/V$ be the density of particles. Under the Debye-H\"uckel (DH) theory, the variance of the random force scales as $\mathcal{O}(1/P)$, which is independent of both the number of particles $N$ and the bandwidths of Gaussians.
\end{theorem}
\begin{proof}
	By the Debye-H\"uckel theory, the term of the structure factor in Eq. \eqref{pro2} can be bounded by a constant  $C$ \cite{jin2020random},  
	\begin{equation}
		\left| \mathrm{Im}\left(e^{-i\boldsymbol{k}\cdot\boldsymbol{r_i}}\cdot\rho(\boldsymbol{k})\right)\right|^2\leq C.
	\end{equation}
	By this inequality, one has 
	\begin{equation}\label{varianceassump}
		\begin{aligned}
			\mathbb{E}|\boldsymbol{\Xi_i}|^2&\le \frac{1}{P}\left\{\frac{4\pi^3q_i^2CS|\bm{k}|^2}{V^2}\cdot\sum_{\bm{k}\neq 0}\widetilde{\mathcal{F}}_{b}^{\sigma}(|\bm{k}|)-\left|\bm{F}_{\mathcal{F},i}\right|^2\right\}
			\\&\lesssim \frac{S}{P}\frac{4 \pi^2q_i^2C}{V^2}\sum_{\ell=0}^{M}\omega_{\ell}s_{\ell}^3\int_{0}^{\infty}\frac{V}{(2\pi)^3}\cdot 4\pi k^2\cdot e^{-\frac{1}{4}s_{\ell}^2 k^2}dk\\
			&=\frac{S}{P}\frac{4\sqrt{\pi}q_i^2C}{V}\sum_{\ell=0}^{M}w_{\ell}.\\
		\end{aligned}
	\end{equation}
	Here, by the definition Eq. \eqref{eq::S},  $S$ has the following estimate: 
	\begin{equation}\label{Sestimate} 
			S=\pi^{-3/2}V\sum_{\ell=0}^{M}\omega_{\ell}\sum_{m_d\in\mathbb{Z}}e^{-\pi^2\sum\limits_{d}m_d^2L_d^2/s_{\ell}^2}-\pi^{-3/2}\widetilde{\mathcal{F}}_{b}^{\sigma}(0), 
	\end{equation}
and thus $S=\mathcal{O}(V)$. 	Substituting it into Eq.\eqref{varianceassump} gives  $\mathbb{E}|\boldsymbol{\Xi_i}|^2=\mathcal{O}(1/P)$, and Eq.  \eqref{varianceassump} clearly shows the independence of the estimate on the particle numbers and the Gaussian bandwidths. 
\end{proof}

We now consider the convergence of the MD. Let $\Delta t$ be the time step of the integration methods (for example, the velocity-Verlet algorithm). Let $(\bm{r}_i, \bm{v}_i)$ be the solution of the underdamped Langevin dynamics equations of motion,
 	\begin{equation}\label{exactforce}
 \begin{split}
 &d\bm{r}_i=\bm{v}_idt,\\
 &m_id\bm{v}_i=\left[\bm{F}_{i}-\gamma\bm{v}_i\right]dt+\sqrt{2\gamma/\beta}d\bm{W}_i,
 \end{split}
 \end{equation}
 where $\{\bm{W}_i\}$ are independently identically distributed Wiener processes and $\bm{r}_i$ and $\bm{v}_i$ denote the coordinates and the velocities of the $i$th particle,  $m_i$ is the mass of the particle, $\gamma$ is the reciprocal characteristic time associated with the thermostat, $\beta=1/k_{\text{B}}T$ is the reciprocal of thermal energy with $k_{\text{B}}$ the Boltzmann constant. Let $(\bm{r}_i^*,\bm{v}_i^*)$ be the solution with the force $\bm{F}_i$ approximated by $\bm{F}_i+\bm{\Xi}_i$ 
 with the same initial data. 
 Let us define the Wasserstein-2 distance \cite{santambrogio2015optimal} as 
 \begin{equation}
 \mathfrak{W}_2(\mu,\nu)=\left(\inf_{\gamma\in\Pi(\mu,\nu)}\int_{\mathbb{R}^3\times\mathbb{R}^3}|\bm{x}-\bm{y}|d\gamma\right)^{1/2},
 \end{equation}
 where $\Pi(\mu,\nu)$ is the adjoint distribution with marginal distributions $\mu$ and $\nu$, respectively.
Theorem \ref{thm:wasser} indicates that RBSOG under the Langevin thermostat \cite{Frenkel2001Understanding} is valid in capturing the finite time structure and dynamic properties. The proof of Theorem \ref{thm:wasser} can be obtained by simply following those in previous work \cite{li2020random,liang2021random2}, and will not present here. 

\begin{theorem}\label{thm:wasser} 
Suppose that the forces $\bm{F}_i$ are bounded and Lipschitz and $\mathbb{E}\bm{\Xi}_i=0$. Let $\bm{Q}$ be the initial configuration of the system, and denote $Y(\bm{Q},\cdot)$ and $Y^*(\bm{Q},\cdot)$ by the transition probabilities of the SDEs driven by the exact force Eq.\eqref{exactforce} and the RBSOG stochastic force, respectively. Then, for any time $t_*>0$, there exists a constant $C(t_*)$ independent of $N$ such that the Wasserstain-2 distance of the two probabilities has the following bound:
	\begin{equation}\label{eq::err}
		\sup_{\boldsymbol{R}}\mathfrak{W}_2(Y,Y^*)\le C(t_*)\sqrt{\Lambda\Delta t+(1+D^2)\Delta t^2}
	\end{equation}
	where $D=\gamma/\beta$ and $\Lambda=\|\mathbb{E}|\bm{\Xi}|^2\|_{\infty}$.
\end{theorem}

\begin{remark}
	Typically, the Nos\'e--Hoover (NH) thermostat \cite{hoover1985canonical} is often adpoted for the heat bath of the NVT and NPT ensembles, instead of the Langevin thermostat. The rigorous proof for the convergence of the NH thermostat remains open, however, we conjecture that a similar error bound as Eq.~\eqref{eq::err} exists considering that the variance term from the random batch sampling will be well controlled by the damping factor in the NH\cite{Jin2020SISC}.
\end{remark}

We analyze the complexity of the RBSOG method for each time step. Given a cutoff radius  $r_c$, the complexity for the near part is certainly linear to the particle number $N$. The random batch sampling Eq.\eqref{eq::approximateForce} for approximating the force results in $P$ terms of the Fourier modes to be summated. In each step, we need to calculate $P$ structure factors $\rho(\bm{k}_\eta)$ for $\eta=1,\cdots, P$, which are used for all particles, thus the complexity for the far part is $\mathcal{O}(PN)$. Furthermore, Theorems \ref{theorem::variance} and \ref{thm:wasser} have indicated that the error of the RBSOG does not  grow with the increase of $N$ for a fixed density and $P=\mathcal{O}(1)$. By these analysis, the RBSOG method has linear complexity per time step.

\section{Numerical results}\label{sec::numerical}
In this section, we perform  all-atom simulations with the RBSOG-based MD under the NVT ensemble to validate both the accuracy and efficiency of the proposed method, with two benchmark systems including  the bulk water and the LiTFSI ionic liquid. In all the RBSOG results,  the SOG decomposition uses the parameters listed in the second row of Table \ref{tabl:parameter} such that the error is at the level of $10^{-4}$. For the near field, the parameters for the core-shell structured tabulation are set as $\mathscr{Q}=4$, $r_{\text{in}}=0.2nm$, $B_{\text{exp}}=3$ and  $B_{\text{man}}=9$.
All the simulations were conducted by our method implemented in the LAMMPS \cite{plimpton1995fast,thompson2022lammps} (version 29Oct2020), and were performed on the ``Siyuan Mark-I'' cluster at Shanghai Jiao Tong University, which comprises $936$ nodes with $2$ $\times$ Intel Xeon ICX Platinum $8358$ CPU ($2.6$ GHz, $32$ cores) and $512$ GB memory per node.  

\subsection{Accuracy for water systems}\label{subsec::water}
We first perform MD simulations on all-atom bulk water system using the extended simple point charge (SPC/E) \cite{berendsen1987missing} force field to examine the accuracy of the RBSOG, compared to the reference PPPM solutions. The system includes $24327$ SPC/E water molecules confined in a cubic box of initial side
length $9nm$. For each case, a short equilibration run of $200$ $ps$ is first conducted, at reference temperature $298$ $K$, with the integration
step size $\Delta t = 1$ $fs$. The relaxation times are set to $\gamma_i = 0.1$ $ps$ for each particle $i$ together with a NH thermostat. The production phase lasts $2$ $ns$, and the configurations are saved every $200$ steps ($0.2$ $ps$) for statistics. The velocity is initially generated according to a Maxwell distribution function. All chemical bonds are converted to constraints using the SHAKE algorithm \cite{krautler2001fast} to allow a time step of $1$ $fs$.  During the equilibration process, the short-range part of the Coulomb interaction of the PPPM and the LJ interaction each with a cutoff parameter of $1nm$ is considered with periodic boundary conditions and the splitting parameter $\alpha=0.26$. The estimated relative error level of the PPPM is about $10^{-4}$ \cite{kolafa1992cutoff} which is consistency with the estimated level of RBSOG. 

\begin{figure*}[!ht]	
	\centering
	\includegraphics[width=1.0\textwidth]{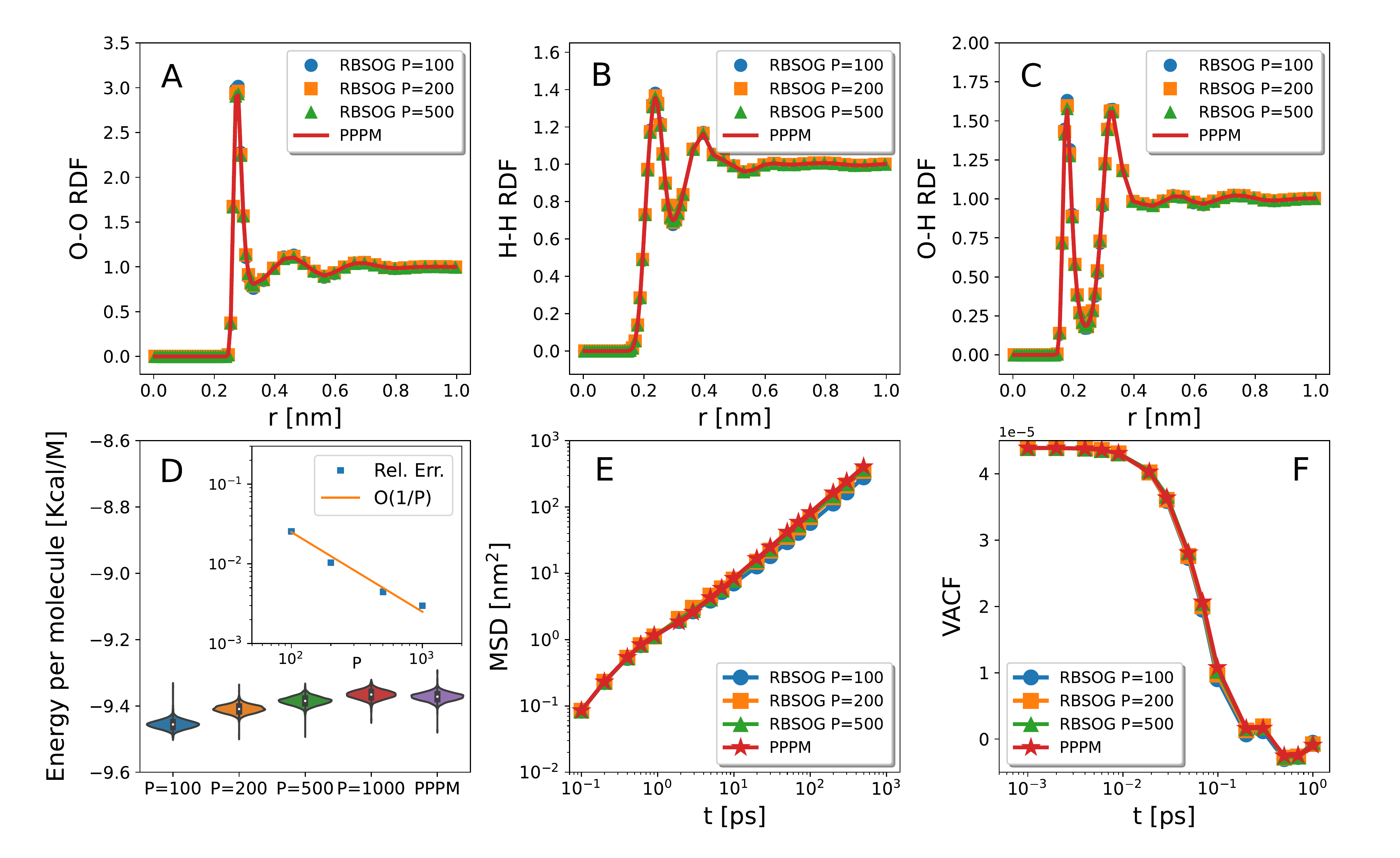}
	\caption{The RDFs of O-O (A), H-H (B), and O-H (C), total energy per molecule (D), MSD (E) and VACF (F) of the bulk water system. The simulation results use the RBSOG method with different batch sizes $P=100, 200$ and $500$, compared to the PPPM. In the violin plot (D), the white point and the two endpoints of black bar within each violin indicate the mean value and two quartiles, respectively. The subplot in (D) shows the convergence on the mean-value errors of the total energy per molecule (including  $P=1000$ data).}
	\label{fig:water1}
\end{figure*}

We measure the properties of the simulation system by the radial distribution function (RDF), the total energy per molecule, the mean square displacement (MSD), and the velocity autocorrelation function (VACF), where the RDF characterizes the equilibrium structure of   water molecules, and the MSD and the VACF are two quantities for measuring the dynamical properties of water. The formulas for these quantities are given in Appendix \ref{appendix::calculation}. The results are shown in Fig. \ref{fig:water1} where panels (A-D) display the RDFs of oxygen-oxygen (O-O), hydrogen-hydrogen (H-H), and oxygen-hydrogen (O-H) and the total energy per molecule. The RBSOG and the PPPM produce statistically identical results on all of the three RDFs. The distributions of Fig.\ref{fig:water1} (D) present the desired Boltzmann distribution, further confirming the accuracy of our RBSOG method. In the subplots of Fig.\ref{fig:water1} (D), the convergence of total energy shows the $\mathcal{O}(P^{-1})$ rate, in consistent with our priori error estimate. Panels  (E) and (F) display the comparisons on the MSD and VACF. The agreement between the RBSOG and the PPPM confirms that dynamical properties are properly reproduced when $P\geq 200$. 


\subsection{Accuracy for ionic liquids}
The second example is the LiTFSI ionic liquid system with the optimized potentials for liquid simulations all-atom (OPLS-AA)  force field for Li$^{+}$ \cite{pronk2013gromacs} and  TFSI$^{-}$ \cite{canongia2004molecular}, and the TIP3P model \cite{price2004modified} for water molecules. This benchmark example was studied by the RBE \cite{liang2022superscalability}.  The system is first equilibrated in the NPT ensemble with the
PPPM at $298$ $K$ and $1$ $bar$ for $500$ $ns$, followed by $200$ $ns$ production MD in the NVT using the NH thermostat with the PPPM
and RBSOG, respectively. The system contains $15803$ atoms, including
$320$ Li$^{+}$, $320$ TFSI$^{-}$, and $3561$ H$_2$O. A cubic simulation box of
size $5.67nm$ is initially used with periodic boundary conditions. The setup for the PPPM is the same as that for the water system. The batch size of the RBSOG takes $P=500$. 

\begin{figure*}[ht]	
	\centering
	\includegraphics[width=0.75\textwidth]{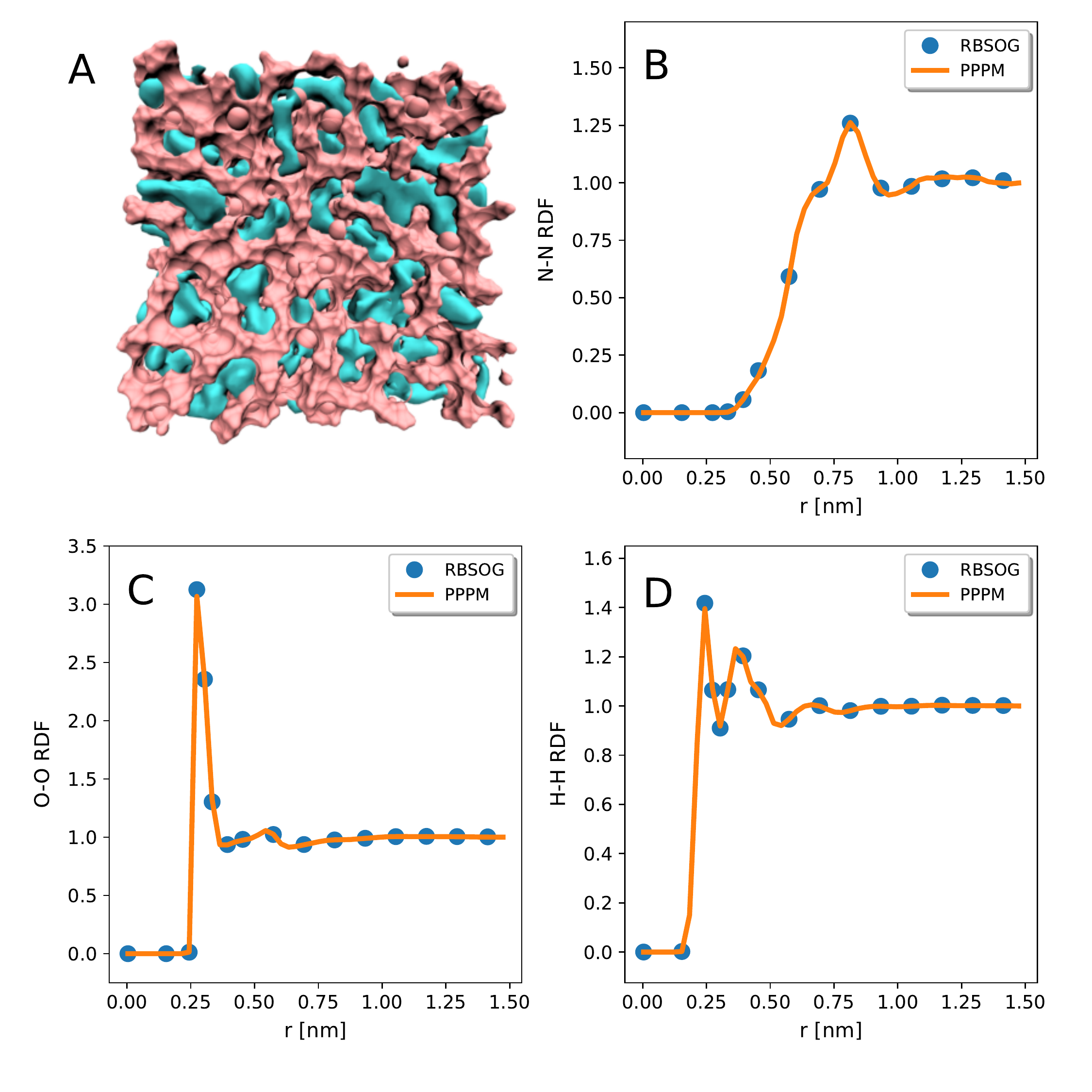}
	\caption{MD simulation results from the RBSOG and the PPPM for the concentrated LiTFSI ionic solution.  (A) MD snapshot of the LiTSFI system; (B-D) The nitrogen-nitrogen, oxygen-oxygen and hydrogen-hydrogen RDFs.}
	\label{fig:rdfmsd}
\end{figure*}

Fig. \ref{fig:rdfmsd}(A) illustrates an MD snapshot for the electrolyte, revealing the nano-heterogeneity of the system. The structural information, i.e., the RDFs of the center atom (nitrogen, oxygen and hydrogen in H$_2$O) of the anions, on this concentrated electrolyte derived from both RBSOG and PPPM are shown in Fig. \ref{fig:rdfmsd}(B-D). The dynamics and fluctuations of the system, including the MSD of oxygen atoms in solvent, total energy, heat conductivity, and viscosity are presented in Fig. \ref{fig:ilmsd} (A-D), respectively. The viscosity is a measure of how viscous a fluid is, and the heat conductivity refers to the ability of liquid to conduct/transfer heat. The calculation methods of the viscosity and heat conductivity are displayed in Appendix \ref{appendix::calculation}. As it can be seen, the spatiotemporal features of the system derived from the two methods are essentially the same. These results indicate that the RBSOG with $P=500$ has comparable accuracy compared with the PPPM at the  $10^{-4}$ error level in this ionic liquid system. 

\begin{figure*}[!ht]	
	\centering
	\includegraphics[width=0.8\textwidth]{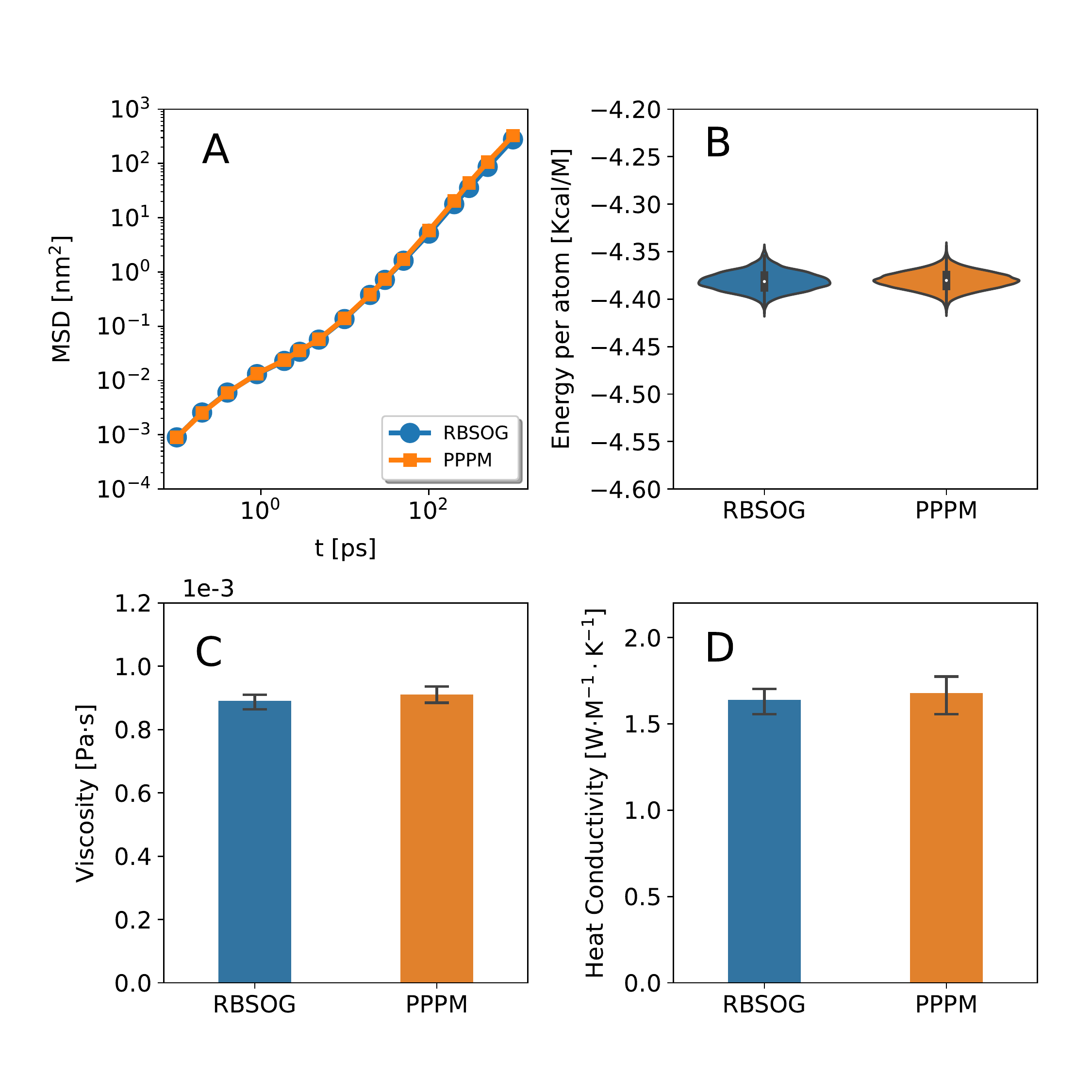}
	\caption{Comparison of dynamical and thermal properties by MD simulations with the PPPM (blue) and RBSOG (pink) for the LiTFSI system.  (A) The O-O MSD of water molecules; (B) the total energy per atom; (C) the thermal conductivity; and (D) the viscosity.}
	\label{fig:ilmsd}
\end{figure*}

\subsection{Time performance of the MD algorithm}

The performance comparisons between the RBSOG and the PPPM were carried out by using LAMMPS on atom simulation of SPC/E pure water systems. To access a fair comparison, the estimated relative force errors is chosen as $10^{-4}$. The parameters of the PPPM are chosen automatically in LAMMPS
based on the error estimates \cite{deserno1998mesh}. The simulations of the system were conducted for $1000$ steps to estimate the average CPU time per step. The density of water molecules is fixed to $1g/cm^3$. The real space cutoff for both Coulomb and LJ potential is set as $r_c=1nm$ for both the RBSOG and PPPM. 

\begin{figure*}[ht]	
	\centering
	\includegraphics[width=0.8\textwidth]{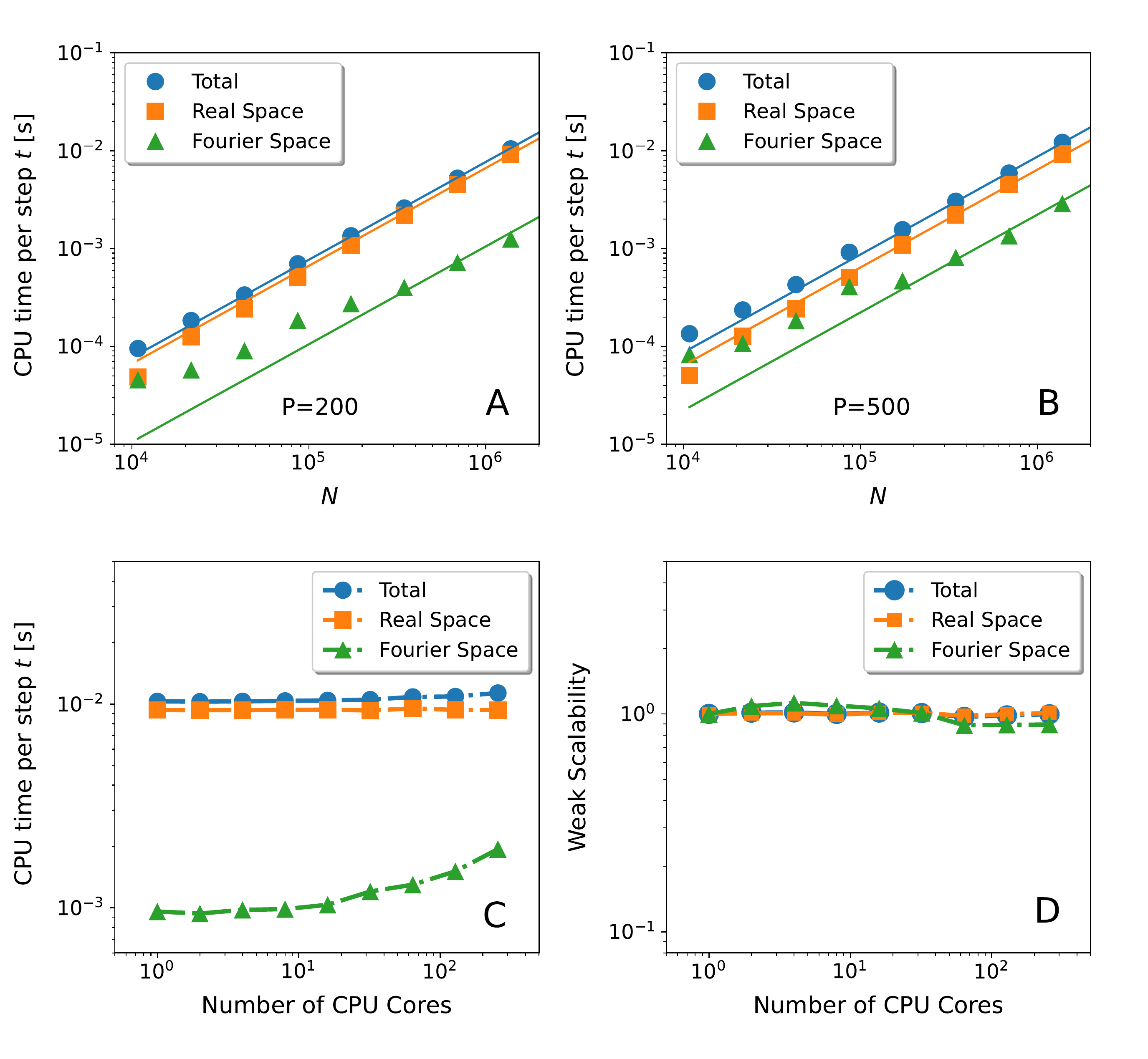}
	\caption{(A-C) CPU time per step for the RBSOG methods with increasing $N$ and (D) weak scalability of the RBSOG methods with  the number of CPU cores. (AB) are the results with 512 cores and with the increase of $N$ for $P=200$ and $500$.  (CD) are the results with the increase of cores and with a fixed average particle  ($2703$ per core).}
	\label{fig:Nscaling}
\end{figure*}

We first characterize the time scaling with increasing number of particles $N$. In Fig. \ref{fig:Nscaling} (A-B), $512$ CPU cores are used, and the computational times for the RBSOG method are shown for system size up to $N=10^6$,  where the solid lines present the linear fitting of the data in log-log scale. The results clearly illustrate   the $\mathcal{O}(N)$ complexity of the RBSOG method. For the FFT-based methods, the proportion of the CPU time of the real space and the Fourier space part are roughly the same. It is observed that the RBSOG has significant computation saving of the Fourier part over a whole range of particle numbers, demonstrating the attractive performance of the algorithm. Note that the first several points of the far part in Fig. \ref{fig:Nscaling} are not linearly scaled, because the number of particles is small and the communication dominates the cost. 


\begin{figure*}[!ht]	
	\centering
	\includegraphics[width=1\textwidth]{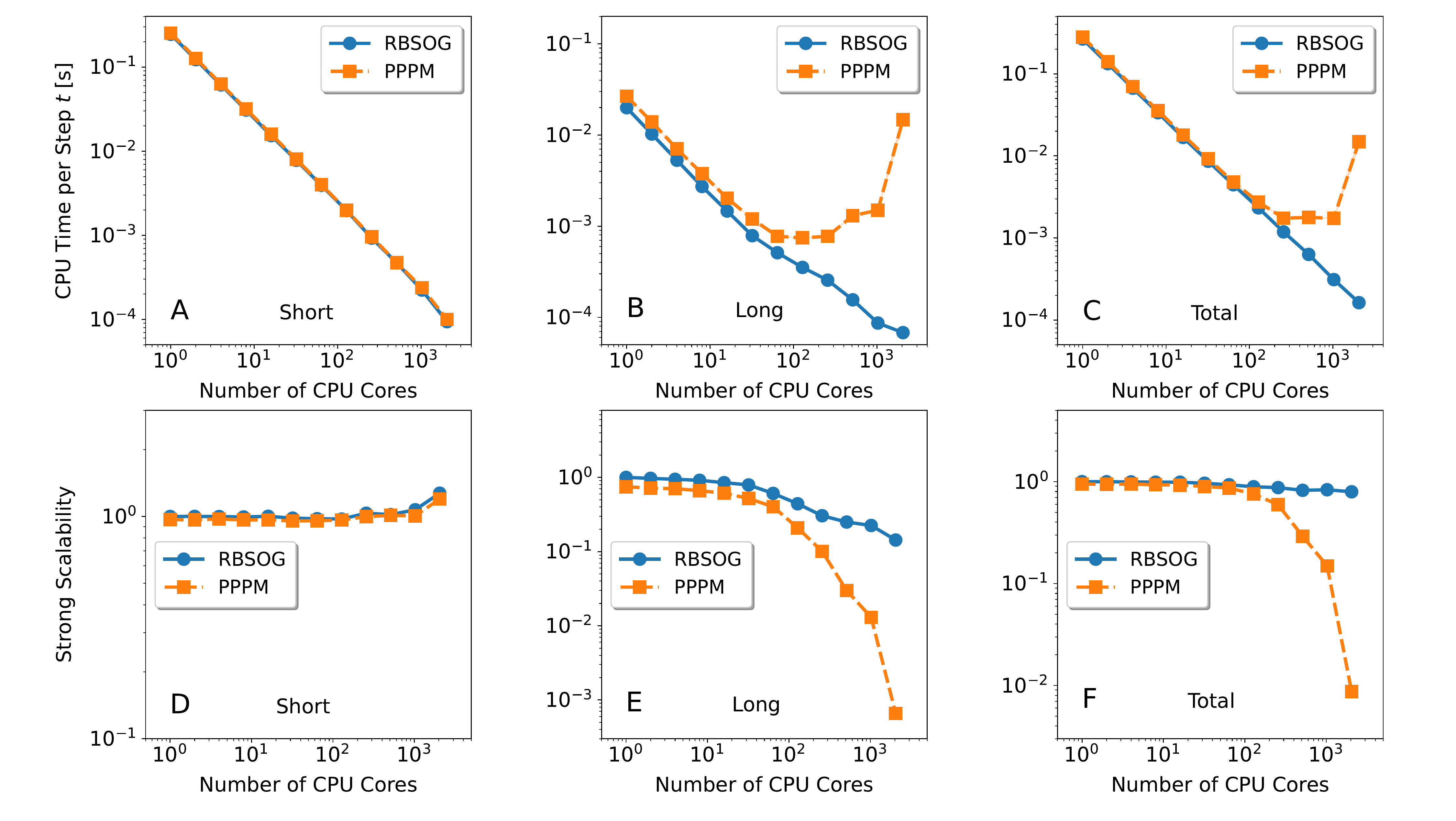}
	\caption{CPU time and strong scalability for comparison between the RBSOG  and the PPPM  with the number of CPU cores up to $2048$. The total number of atoms is $72981$.  Data are shown for CPU time per step (A-C) and strong scalability (D-F) of the near part (short), the  far part (long), and the total force (total). }
	\label{fig:Strongscaling}
\end{figure*}

The algorithm scalability is one of key issues limiting both system scale and time scale of MD simulations. It is well known that the FFT for evaluating Coulomb interactions requires intensive communication costs, which asymptotically takes more than $95\%$ of
runtime on hybrid systems \cite{ayala2021scalability,Arnold2013PRE}. Hence, it is critical to qualify a novel Coulomb solver by examining how the scalability can be improved. In this work, we
measure the scalability of the RBSOG by both the weak scaling and the strong scaling. The weak scaling characterizes the parallel performance tuning the number of processors (and the system size) by fixing a average number of particles per processor. The Gustafson's law \cite{rf31reevaluating} indicates that an algorithm with a perfectly weak scaling has a speedup of $1$ and stationary running time for large scale systems. Alternatively, the strong scaling   remains the system size, and tunes  the number of CPU cores. Let $\eta$ be the core number and $T(\eta)$ is the corresponding run time.  The strong scaling  is defined by
\begin{equation}\label{eq:relativescalability}
	\eta(\lambda)=\dfrac{\lambda_{\min}}{\lambda}\cdot\dfrac{T_{\text{min}}}{T(\lambda)},
\end{equation}
where $T_{\text{min}}=T(\lambda_{\min})$ and $\lambda_{\min}$ is the minimal number of cores in the calculation.

The results of weak scalability on massive high-performance cluster with up to about one thousand cores are present in Fig.~\ref{fig:Nscaling}(CD). The weak scalability and the corresponding average CPU time per step show that the RBSOG achieves near-perfect weak scaling for both real space and Fourier space parts, as well as the total (i.e., real+Fourier) CPU cost, in spite that the average number of particles on one core is relatively small ($2703$ particles per core). The weak scaling of the Fourier component slightly deviates from 1 due to the increasing of the ratio between the communication cost to the total cost. Fig.\ref{fig:Strongscaling} (A-F) show the CPU performance and strong scaling of the simulation results by the RBSOG and the PPPM. When the core number is small, the RBSOG has almost the same performance as the PPPM. The RBSOG outperforms the PPPM with the increase of the core number, and the advangtage becomes significant for $\lambda>100$.  When   $\lambda=2048$, the total computational speed of the RBSOG is two order of magnitude faster than that of the PPPM. The strong scalability of the RBSOG remains over $90\%$ when $2048$ CPU cores are employed, significantly outperforming that of the PPPM which drops to ~ $1\%$ for the same system. These results demonstrate the great performance of the RBSOG in parallel scalability. It is mentionable that the RBSOG has $\sim 5\%$ improvement in the near-part calculation over that of the PPPM. This owes to the core-shell structured tabulation technique.

\section{Concluding remarks}\label{sec::conclusion}

We have developed a novel RBSOG algorithm that is accurate and efficient for the MD simulations with long-range Coulomb interactions which requires $\mathcal{O}(N)$ operations each step and less communication costs. The RBSOG is based on the SOG decomposition which splits the Coulomb kernel as two parts. One is an SOG which is long-ranged, and another is $1/r$ minus the SOG which is short-ranged. The resulting decomposition is smooth on the entire real axis, and is exact up to a cutoff radius $r_c$. The RBSOG builds the random mini-batch strategy into the Fourier space for the long-range forces, together with an importance sampling for the Fourier modes, so that it takes $\mathcal{O}(N)$ operations per time step. Discussions on infinite boundary term, error analysis, and implementation are provided. The all-atom simulations on bulk water and ionic liquid systems show that the RBSOG algorithm can quantitatively reproduce the spatiotemporal information and thermodynamic quantities, and shows attractive performance regarding the efficiency and scalability on massive supercomputer cluster. 

We perform the comparison of the RBSOG results with the PPPM, and demonstrating the advantages of the algorithm. It is noted that the RBE  has a similar performance as the RBSOG, and the systematical comparison between the RBE and the PPPM was reported \cite{liang2022superscalability}. We remark that in comparison with the Ewald splitting, the SOG decomposition has a better smoothness near the cutoff, and thus the truncation error can be reduced. Also, the force valence by the random batch sampling may be different between the two methods, and the RBSOG may achieve a smaller variance due to the possible use of a bigger bandwidth lowerbound as was pointed \cite{DEShaw2020JCP}. The investigation on these issues from both computational and theoretical points of view shall be conducted in the future. 

The RBSOG can be easily extended to other kernels. Different from the bilateral series approximation, one can achieve it by developing kernel-independent SOG methods such that the Gaussian bandwidths have controllable upperbound.  Incorporating such a SOG technique for MD simulations with non-Coulomb long-range kernels is also the goal of our subsequent works. 

\section*{Acknowledgements}
The authors acknowledge the financial support from the National Natural Science Foundation of China (grant No. 12071288), Science and Technology Commission of Shanghai Municipality (grant Nos. 20JC1414100 and 21JC1403700), and the support from the HPC center of Shanghai Jiao Tong University.

\appendix 
\section*{Appendix}

\section{Physical quantities calculated in  results}\label{appendix::calculation}
The RDF $g_{\text{rdf}}(r)$ at distance $r$ is defined by
\begin{equation}
	g_{\text{rdf}}=\dfrac{1}{N\rho_r}\sum_{i=1}^{N}\sum_{j\neq i}\dfrac{\langle\delta(r_{ij}-r)\rangle}{4\pi r^2}
\end{equation}
where $\delta(\cdot)$ is the Dirac delta function and the bracket represents the ensemble average.
The MSD $\eta_{\text{msd}}(t)$ at time $t$ is defined as
\begin{equation}
	\eta_{\text{msd}}(t)=\langle|\bm{r}(t+t_0)-\bm{r}(t_0)|^2\rangle
\end{equation}
with the bracket representing the ensemble average over $t_0$.
The VACF $\eta_{vacf}(t)$ at time $t$ is defined as
\begin{equation}
	\eta_{vacf}(t)=\langle\bm{v}(t_0)\bm{v}(t)\rangle
\end{equation}
with the bracket representing the ensemble average over $t_0$.
The total energy is defined as the sum of potential energy and kinetic energy. The potential energy is the sum of bond, angle, dihedral, improper, Coulomb, LJ and constrain components.

The viscosity is calculated by using the Green-Kubo relation
\begin{equation}
	C_{\eta}=\dfrac{V}{6k_{\text{B}}T}\sum_{\alpha\leq\beta}\int_{0}^{\infty}\langle\overline{P}_{\alpha\beta}(t)\cdot \overline{P}_{\alpha\beta}(0)\rangle dt
\end{equation}
where $V$ is the system volume, $T$ is the temperature, and $\overline{P}_{\alpha\beta}$ denotes an element $\alpha\beta$ of pressure tensor with $\alpha,\beta\in\{1,2,3\}$. Here the pressure is stored as a 6-element tensor, and is defined as follows:
\begin{equation}\label{eqeq::pressure}
	\overline{P}_{\alpha\beta}=\frac{1}{V} 
	 \sum_{i=1}^{N}( m_{i}v_{i,\alpha}v_{i,\beta} +  r_{i,\alpha}f_{i,\beta} ) 
\end{equation}
where $v_{i,\alpha}$, $r_{i,\alpha}$, and $f_{i,\alpha}$ are the $\alpha$-th component of the velocity $\bm{v}$, the position $\bm{r}$, and the force $\bm{f}$ of particle $i$, respectively. The two components in each term come from  the kenetic energy and the virial contributions, respectively. 
The heat conductivity, which is related to the ensemble average of the auto-correlation of the heat flux $\bm{J}$, is given via the Green-Kubo formula
\begin{equation}
	C_{\kappa}=\dfrac{V}{3k_{\text{B}}T^2}\int_{0}^{\infty}\langle \bm{J}(0)\cdot \bm{J}(t) \rangle dt.
\end{equation}
Here the heat flux $\bm{J}$ is defined as
\begin{equation}
	\bm{J}=\dfrac{1}{V}\left[\sum_{i=1}^{N}\overline{e}_{i}\bm{v}_i+\dfrac{1}{2}\sum_{\substack{i,j=1\\i<j}}^{N}(\bm{f}_{ij}\cdot(\bm{v}_i+\bm{v}_j))\bm{r}_{ij}\right]
\end{equation}
where $\overline{e}_{i}$ is the per-atom energy (potential and kinetic).


\begin{thebibliography}{10}
	
	\bibitem{abraham2015gromacs}
	{\sc M.~J. Abraham, T.~Murtola, R.~Schulz, S.~P{\'a}ll, J.~C. Smith, B.~Hess,
		and E.~Lindahl}, {\em {GROMACS: High performance molecular simulations
			through multi-level parallelism from laptops to supercomputers}}, SoftwareX,
	1 (2015), pp.~19--25.
	
	\bibitem{Allen2017ComputerLiquids}
	{\sc M.~P. Allen and D.~J. Tildesley}, {\em {Computer Simulation of Liquids}},
	Oxford University Press, 2017.
	
	\bibitem{andrea1983role}
	{\sc T.~A. Andrea, W.~C. Swope, and H.~C. Andersen}, {\em The role of long
		ranged forces in determining the structure and properties of liquid water},
	J. Chem. Phys., 79 (1983), pp.~4576--4584.
	
	\bibitem{Arnold2013PRE}
	{\sc A.~Arnold, F.~Fahrenberger, C.~Holm, O.~Lenz, M.~Bolten, H.~Dachsel,
		R.~Halver, I.~Kabadshow, F.~G\"ahler, and F.~Heber}, {\em Comparison of
		scalable fast methods for long-range interactions}, Phys. Rev. E, 88 (2013),
	p.~063308.
	
	\bibitem{ayala2021scalability}
	{\sc A.~Ayala, S.~Tomov, M.~Stoyanov, and J.~Dongarra}, {\em {Scalability
			issues in FFT computation}}, in International Conference on Parallel
	Computing Technologies, Springer, 2021, pp.~279--287.
	
	\bibitem{Barnes1986Nature}
	{\sc J.~Barnes and P.~Hut}, {\em {A hierarchical O(NlogN) force-calculation
			algorithm}}, Nature, 324 (1986), pp.~446--449.
	
	\bibitem{berendsen1987missing}
	{\sc H.~Berendsen, J.~Grigera, and T.~Straatsma}, {\em The missing term in
		effective pair potentials}, J. Phys. Chem., 91 (1987), pp.~6269--6271.
	
	\bibitem{beylkin2005approximation}
	{\sc G.~Beylkin and L.~Monz{\'o}n}, {\em On approximation of functions by
		exponential sums}, Appl. Comput. Harmon. Anal., 19 (2005), pp.~17--48.
	
	\bibitem{beylkin2010approximation}
	\leavevmode\vrule height 2pt depth -1.6pt width 23pt, {\em Approximation by
		exponential sums revisited}, Appl. Comput. Harmon. Anal., 28 (2010),
	pp.~131--149.
	
	\bibitem{canongia2004molecular}
	{\sc J.~N. Canongia~Lopes and A.~A. P{\'a}dua}, {\em Molecular force field for
		ionic liquids composed of triflate or bistriflylimide anions}, J. Phys. Chem.
	B, 108 (2004), pp.~16893--16898.
	
	\bibitem{cheng1999fast}
	{\sc H.~Cheng, L.~Greengard, and V.~Rokhlin}, {\em A fast adaptive multipole
		algorithm in three dimensions}, J. Comput. Phys., 155 (1999), pp.~468--498.
	
	\bibitem{Darden1993JCP}
	{\sc T.~Darden, D.~York, and L.~Pedersen}, {\em {Particle mesh Ewald: An
			N$\cdot$log(N) method for Ewald sums in large systems}}, J. Chem. Phys., 98
	(1993), pp.~10089--10092.
	
	\bibitem{deserno1998mesh}
	{\sc M.~Deserno and C.~Holm}, {\em {How to mesh up Ewald sums. II. An accurate
			error estimate for the particle--particle--particle-mesh algorithm}}, J.
	Chem. Phys., 109 (1998), pp.~7694--7701.
	
	\bibitem{dos2016simulations}
	{\sc A.~P. dos Santos, M.~Girotto, and Y.~Levin}, {\em {Simulations of Coulomb
			systems with slab geometry using an efficient 3D Ewald summation method}}, J.
	Chem. Phys., 144 (2016), p.~144103.
	
	\bibitem{essmann1995smooth}
	{\sc U.~Essmann, L.~Perera, M.~L. Berkowitz, T.~Darden, H.~Lee, and L.~G.
		Pedersen}, {\em {A smooth particle mesh Ewald method}}, J. Chem. Phys., 103
	(1995), pp.~8577--8593.
	
	\bibitem{Ewald1921AnnPhys}
	{\sc P.~P. Ewald}, {\em {Die Berechnung optischer und elektrostatischer
			Gitterpotentiale}}, Ann. Phys., 369 (1921), pp.~253--287.
	
	\bibitem{exl2016accurate}
	{\sc L.~Exl, N.~J. Mauser, and Y.~Zhang}, {\em {Accurate and efficient
			computation of nonlocal potentials based on Gaussian-sum approximation}}, J.
	Comput. Phys., 327 (2016), pp.~629--642.
	
	\bibitem{French2010Rev}
	{\sc R.~H. French, V.~A. Parsegian, R.~Podgornik, R.~F. Rajter, A.~Jagota,
		J.~Luo, D.~Asthagiri, M.~K. Chaudhury, Y.-m. Chiang, and S.~Granick}, {\em
		Long range interactions in nanoscale science}, Rev. Mod. Phys., 82 (2010),
	p.~1887.
	
	\bibitem{Frenkel2001Understanding}
	{\sc D.~Frenkel and B.~Smit}, {\em {Understanding Molecular Simulation: From
			Algorithms to Applications}}, vol.~1, Elsevier, 2001.
	
	\bibitem{greengard2018anisotropic}
	{\sc L.~Greengard, S.~Jiang, and Y.~Zhang}, {\em The anisotropic truncated
		kernel method for convolution with free-space {Green's} functions}, SIAM J.
	Sci. Comput., 40 (2018), pp.~A3733--A3754.
	
	\bibitem{greengard1987fast}
	{\sc L.~Greengard and V.~Rokhlin}, {\em A fast algorithm for particle
		simulations}, {J. Comput. Phys.}, 73 (1987), pp.~325--348.
	
	\bibitem{rf31reevaluating}
	{\sc J.~L. Gustafson}, {\em {Reevaluating Amdhal's law}}, Comm. ACM, 31,
	pp.~532--533.
	
	\bibitem{hastings1970monte}
	{\sc W.~K. Hastings}, {\em {Monte Carlo sampling methods using Markov chains
			and their applications}},  (1970).
	
	\bibitem{Hockney1988Computer}
	{\sc R.~W. Hockney and J.~W. Eastwood}, {\em {Computer Simulation Using
			Particles}}, CRC Press, 1988.
	
	\bibitem{hoover1985canonical}
	{\sc W.~G. Hoover}, {\em {Canonical dynamics: Equilibrium phase-space
			distributions}}, Phys. Rev. A, 31 (1985), p.~1695.
	
	\bibitem{hu2014infinite}
	{\sc Z.~Hu}, {\em {Infinite boundary terms of Ewald sums and pairwise
			interactions for electrostatics in bulk and at interfaces}}, J. Chem. Theory
	Comput., 10 (2014), pp.~5254--5264.
	
	\bibitem{jiang2008efficient}
	{\sc S.~Jiang and L.~Greengard}, {\em {Efficient representation of
			nonreflecting boundary conditions for the time-dependent Schr{\"o}dinger
			equation in two dimensions}}, Commun. Pur. Appl. Math., 61 (2008),
	pp.~261--288.
	
	\bibitem{jin2020random}
	{\sc S.~Jin, L.~Li, and J.-G. Liu}, {\em {Random batch methods (RBM) for
			interacting particle systems}}, J. Comput. Phys., 400 (2020), p.~108877.
	
	\bibitem{Jin2020SISC}
	{\sc S.~Jin, L.~Li, Z.~Xu, and Y.~Zhao}, {\em {A random batch Ewald method for
			particle systems with Coulomb interactions}}, SIAM J. Sci. Comput., 43
	(2021), pp.~B937--B960.
	
	\bibitem{kolafa1992cutoff}
	{\sc J.~Kolafa and J.~W. Perram}, {\em {Cutoff errors in the Ewald summation
			formulae for point charge systems}}, Mol. Simulat., 9 (1992), pp.~351--368.
	
	\bibitem{krautler2001fast}
	{\sc V.~Kr{\"a}utler, W.~F. Van~Gunsteren, and P.~H. H{\"u}nenberger}, {\em {A
			fast SHAKE algorithm to solve distance constraint equations for small
			molecules in molecular dynamics simulations}}, J. Comput. Chem., 22 (2001),
	pp.~501--508.
	
	\bibitem{leiserson2020there}
	{\sc C.~E. Leiserson, N.~C. Thompson, J.~S. Emer, B.~C. Kuszmaul, B.~W.
		Lampson, D.~Sanchez, and T.~B. Schardl}, {\em {There's plenty of room at the
			Top: What will drive computer performance after Moore's law?}}, Science, 368
	(2020), p.~eaam9744.
	
	\bibitem{li2020random}
	{\sc L.~Li, Z.~Xu, and Y.~Zhao}, {\em {A random-batch Monte Carlo method for
			many-body systems with singular kernels}}, SIAM J. Sci. Comput., 42 (2020),
	pp.~A1486--A1509.
	
	\bibitem{jiuyang2021AAMM}
	{\sc J.~Liang, Z.~Gao, and Z.~Xu}, {\em {A kernel-independent sum-of-Gaussians
			method by de la Vallee-Poussin sums}}, Adv. Appl. Math. Mech., 13 (2021),
	pp.~1126--1141.
	
	\bibitem{liang2021random}
	{\sc J.~Liang, P.~Tan, L.~Hong, S.~Jin, Z.~Xu, and L.~Li}, {\em {A random batch
			Ewald method for charged particles in the isothermal-isobaric ensemble}},
	arXiv:2110.14362,  (2021).
	
	\bibitem{liang2022superscalability}
	{\sc J.~Liang, P.~Tan, Y.~Zhao, L.~Li, S.~Jin, L.~Hong, and Z.~Xu}, {\em
		{Superscalability of the random batch Ewald method}}, J. Chem. Phys., 156
	(2022), p.~014114.
	
	\bibitem{liang2021random2}
	{\sc J.~Liang, Z.~Xu, and Y.~Zhao}, {\em Random-batch list algorithm for
		short-range molecular dynamics simulations}, J. Chem. Phys., 155 (2021),
	p.~044108.
	
	\bibitem{liang2020harmonic}
	{\sc J.~Liang, J.~Yuan, E.~Luijten, and Z.~Xu}, {\em Harmonic surface mapping
		algorithm for molecular dynamics simulations of particle systems with planar
		dielectric interfaces}, J. Chem. Phys., 152 (2020), p.~134109.
	
	\bibitem{maggs2002local}
	{\sc A.~Maggs and V.~Rossetto}, {\em {Local simulation algorithms for Coulomb
			interactions}}, Phys. Rev. Lett., 88 (2002), p.~196402.
	
	\bibitem{maxian2021fast}
	{\sc O.~Maxian, R.~P. Pel{\'a}ez, L.~Greengard, and A.~Donev}, {\em A fast
		spectral method for electrostatics in doubly periodic slit channels}, J.
	Chem. Phys., 154 (2021), p.~204107.
	
	\bibitem{parola1995liquid}
	{\sc A.~Parola and L.~Reatto}, {\em Liquid state theories and critical
		phenomena}, Adv. Phys., 44 (1995), pp.~211--298.
	
	\bibitem{plimpton1995fast}
	{\sc S.~Plimpton}, {\em Fast parallel algorithms for short-range molecular
		dynamics}, J. Comput. Phys., 117 (1995), pp.~1--19.
	
	\bibitem{DEShaw2020JCP}
	{\sc C.~Predescu, A.~K. Lerer, R.~A. Lippert, B.~Towles, J.~Grossman, R.~M.
		Dirks, and D.~E. Shaw}, {\em {The u-series: A separable decomposition for
			electrostatics computation with improved accuracy}}, J. Chem. Phys., 152
	(2020), p.~084113.
	
	\bibitem{price2004modified}
	{\sc D.~J. Price and C.~L. Brooks~III}, {\em {A modified TIP3P water potential
			for simulation with Ewald summation}}, J. Chem. Phys., 121 (2004),
	pp.~10096--10103.
	
	\bibitem{pronk2013gromacs}
	{\sc S.~Pronk, S.~P{\'a}ll, R.~Schulz, P.~Larsson, P.~Bjelkmar, R.~Apostolov,
		M.~R. Shirts, J.~C. Smith, P.~M. Kasson, D.~Van Der~Spoel, et~al.}, {\em
		{GROMACS 4.5: a high-throughput and highly parallel open source molecular
			simulation toolkit}}, Bioinformatics, 29 (2013), pp.~845--854.
	
	\bibitem{santambrogio2015optimal}
	{\sc F.~Santambrogio}, {\em Optimal transport for applied mathematicians},
	Birk{\"a}user, NY, 55 (2015), p.~94.
	
	\bibitem{shan2005gaussian}
	{\sc Y.~Shan, J.~L. Klepeis, M.~P. Eastwood, R.~O. Dror, and D.~E. Shaw}, {\em
		{Gaussian split Ewald: A fast Ewald mesh method for molecular simulation}},
	J. Chem. Phys., 122 (2005), p.~054101.
	
	\bibitem{shaw2021anton}
	{\sc D.~E. Shaw, P.~J. Adams, A.~Azaria, J.~A. Bank, B.~Batson, A.~Bell,
		M.~Bergdorf, J.~Bhatt, J.~A. Butts, T.~Correia, et~al.}, {\em {Anton 3:
			twenty microseconds of molecular dynamics simulation before lunch}}, in
	Proceedings of the International Conference for High Performance Computing,
	Networking, Storage and Analysis, 2021, pp.~1--11.
	
	\bibitem{smith1981electrostatic}
	{\sc E.~R. Smith}, {\em Electrostatic energy in ionic crystals}, Proc. R. Soc.
	London, 375 (1981), pp.~475--505.
	
	\bibitem{thompson2022lammps}
	{\sc A.~P. Thompson, H.~M. Aktulga, R.~Berger, D.~S. Bolintineanu, W.~M. Brown,
		P.~S. Crozier, P.~J. in't Veld, A.~Kohlmeyer, S.~G. Moore, T.~D. Nguyen,
		et~al.}, {\em {LAMMPS-a flexible simulation tool for particle-based materials
			modeling at the atomic, meso, and continuum scales}}, Comput. Phys. Commun.,
	271 (2022), p.~108171.
	
	\bibitem{trottenberg2000multigrid}
	{\sc U.~Trottenberg, C.~W. Oosterlee, and A.~Schuller}, {\em Multigrid},
	Elsevier, 2000.
	
	\bibitem{David2011Nanoscale}
	{\sc D.~A. Walker, B.~Kowalczyk, M.~O. de~la Cruz, and B.~A. Grzybowski}, {\em
		Electrostatics at the nanoscale}, Nanoscale, 3 (2011), pp.~1316--1344.
	
	\bibitem{weeks1971role}
	{\sc J.~D. Weeks, D.~Chandler, and H.~C. Andersen}, {\em Role of repulsive
		forces in determining the equilibrium structure of simple liquids}, J. Chem.
	Phys., 54 (1971), pp.~5237--5247.
	
	\bibitem{wiscombe1977exponential}
	{\sc W.~Wiscombe and J.~Evans}, {\em Exponential-sum fitting of radiative
		transmission functions}, J. Comput. Phys., 24 (1977), pp.~416--444.
	
	\bibitem{wolff1999tabulated}
	{\sc D.~Wolff and W.~Rudd}, {\em Tabulated potentials in molecular dynamics
		simulations}, Comput. Phys. Commun., 120 (1999), pp.~20--32.
	
	\bibitem{yeh1999ewald}
	{\sc I.-C. Yeh and M.~L. Berkowitz}, {\em Ewald summation for systems with slab
		geometry}, J. Chem. Phys., 111 (1999), pp.~3155--3162.
	
\end{thebibliography}

\end{document}